\documentclass[runningheads, envcountsame]{llncs}

\usepackage[T1]{fontenc}
\usepackage{graphicx}
\usepackage{amssymb}
\usepackage{amsmath}
\usepackage{todonotes}
\usepackage{paralist}
\usepackage{xspace}
\usepackage[ruled, vlined, linesnumbered]{algorithm2e}
\usepackage[bookmarks=true, bookmarksnumbered=true, bookmarksopen=false, bookmarksdepth=2]{hyperref}
\usepackage{bookmark}

\usetikzlibrary{arrows.meta, math, decorations.markings}
\usepackage{color}

\newcommand{\sep}{\;|\;}
\newcommand{\seq}{\subseteq}

\newcommand{\ca}{\mathcal}

\newtheorem{observation}[theorem]{Observation}
\newtheorem{rr}[theorem]{Reduction rule}

\DeclareMathOperator{\dom}{dom}
\DeclareMathOperator{\range}{range}
\DeclareMathOperator{\dist}{dist}
\DeclareMathOperator{\girth}{girth}
\DeclareMathOperator{\td}{td}
\DeclareMathOperator{\pw}{pw}
\DeclareMathOperator{\mw}{mw}
\DeclareMathOperator{\tww}{tww}
\DeclareMathOperator{\tw}{tw}

\newcommand{\FPT}{\textsf{FPT}\xspace}
\newcommand{\NP}{\textsf{NP}\xspace}
\newcommand{\XP}{\textsf{XP}\xspace}
\newcommand{\XNLP}{\textsf{XNLP}\xspace}
\newcommand{\Wone}{\textsf{W[1]}\xspace}
\newcommand{\paraNP}{\textsf{para-NP}\xspace}

\newcommand\cnone{\mathsf{none}}
\newcommand\ccontains{\mathsf{cont}}
\newcommand\cdemand{\mathsf{dem}}
\newcommand\ctypes{\mathsf{types}}
\newcommand\csig{\mathsf{sig}}
\newcommand\tab{\mathsf{tab}}

\newcommand{\bscoloring}{$b^*$-\textsc{coloring}\xspace}
\newcommand{\bcoloring}{$b$-\textsc{coloring}\xspace}
\newcommand{\bchromatic}{\textsc{$b$-chromatic number}\xspace}

\begin{document}

\title{From $b$-Coloring to $b^*$-Coloring: Large Girth and Parameterized Complexity}
%
%
\author{Jakub Balabán\orcidID{0000-0002-2475-8938}\thanks{Brno Ph.D. Talent Scholarship Holder – Funded by the Brno City Municipality}
 \and
Oliver Bukor}

\authorrunning{J. Balabán and O. Bukor}

\institute{Faculty of Informatics, Masaryk University, Botanická 68a, Brno, Czech Republic
\email{jakbal@mail.muni.cz, 525087@mail.muni.cz}}

\maketitle

\begin{abstract}

A $b$-coloring is a proper vertex coloring such that every color class contains a vertex, a so-called $b$-vertex, which sees all colors in its closed neighborhood.
This type of coloring has been intensively studied from both structural and algorithmic point of view.
Recently, Zaker~[DAM 2025] introduced the notion of a $b^*$-coloring, which is a $b$-coloring in which there is a vertex that sees a $b$-vertex of every color in its closed neighborhood.
The $b^*$-chromatic number is the maximum integer $k$ such that there is a $b^*$-coloring with $k$ colors.

We partially answer a question posed by Zaker and prove that graphs of girth at least 7 are $b^*$-monotonic, which means that the $b^*$-chromatic number does not increase by taking an induced subgraph.
In addition, we discover a class of $d$-regular graphs of girth at least 5 with $b^*$-chromatic number $d+1$, which strengthens a result about $b$-colorings by Dettlaff, Furmańczyk, Peterin, Roux, and Ziemann~[AMC 2024].

We also study the parameterized complexity of finding $b^*$-colorings, and show that for many structural parameters, the complexity coincides with that of finding $b$-colorings.
In particular, the $b^*$-chromatic number can be computed in polynomial time on any class of bounded clique-width.
For most parameters, the translation from $b$-colorings is straightforward but for the feedback edge number, the \FPT algorithm for $b^*$-colorings is actually much simpler than that for $b$-colorings by Balabán~[MFCS 2026].

\end{abstract}

\section{Introduction}

In 1999, Irving and Manlove~\cite{IRVING1999127} defined a coloring\footnote{All colorings considered in this paper are proper, i.e., no two adjacent vertices are allowed to receive the same color.} to be a \emph{$b$-coloring} if every color class contains a vertex, a so-called $b$-vertex, which has a neighbor in all other color classes.
Observe that given a coloring that is \emph{not} a $b$-coloring, we can compute a coloring with fewer colors by (1) finding a color class $C$ with no $b$-vertex, and (2) recoloring every vertex $u \in C$ to some color not present in the neighborhood of $u$.
If we start with an arbitrary coloring and repeat this procedure as long as possible, we eventually reach a $b$-coloring.
This process can be viewed as a simple heuristic for finding a coloring with as few colors as possible: let us call it the $b$-heuristic.
The \emph{$b$-chromatic number} $b(G)$ of a graph $G$, which is the largest integer $k$ such that there is a $b$-coloring of $G$ with $k$ colors, describes the worst-case behavior of the $b$-heuristic.

In the \bchromatic problem, the task is to decide, given a graph $G$ and an integer $k$, whether $b(G) \ge k$.
This problem is \NP-complete~\cite{IRVING1999127}, even if $G$ is bipartite~\cite{KratochvilTV02}, chordal~\cite{HavetSS12}, co-bipartite~\cite{BonomoSSV15}, or a line graph~\cite{CamposLMSSS15}.
On the other hand, \bchromatic is in \textsf{P} on trees~\cite{IRVING1999127} and other tree-like classes such as tree-cographs~\cite{BonomoSSV15}, cacti~\cite{CamposSMS09} and graphs of girth at least 7~\cite{CamposLS15}.
Most of these algorithmic results were unified by the \XP\footnote{For definition of parameterized complexity classes such as \FPT and \XP, see~\cite{flum2006parameterized}. Structural graph parameters used in this paper are defined in Section~\ref{sub:gp}.} algorithm parameterized by the clique-width of $G$ for the \bcoloring problem~\cite{JaffkeLL24}, in which the task is to decide whether $G$ admits a $b$-coloring with exactly $k$ colors.
We remark that a polytime algorithm for \bcoloring can be simply used to solve \bchromatic in polynomial time but not vice versa because the set $\{k \in \mathbb N \sep G$ admits a $b$-coloring with $k$ colors$\}$ does not have to form an interval\footnote{For example, the bipartite graph with vertex set $\{u_1, u_2, u_3, u_4, v_1, v_2, v_3, v_4\}$ and edge set $\{u_iv_j \sep i \ne j\}$ admits a $b$-coloring with two or four colors but not with three colors~\cite{IRVING1999127}.}.
In addition to these algorithmic results, $b$-colorings have also been studied in the context of regular graphs~\cite{KratochvilTV02,CabelloJ11,SahiliKM15,jakovac2010b,shaebani2012b,el2014b}, Kneser graphs~\cite{javadi2009b,hajiabolhassan2010b,balakrishnan2012b,shaebani2019note}, and graph products~\cite{kouider2007b,balakrishnan2014b,maffray2013b,jakovac2012b}.
Finally, several notions related to $b$-colorings, such as $b$-perfectness~\cite{hoang2005b}, $b$-continuity~\cite{barth2007b}, $b$-criticality~\cite{ikhlef2010characterization}, and $b$-monotonicity~\cite{bonomo2009b}, have been introduced.
See the following survey~\cite{jakovac2018b} for more information on $b$-colorings.

In 2025, Zaker~\cite{ZAKER2025370} defined a coloring to be a $b^*$-\emph{coloring} if there is a vertex, a so-called $b^*$-vertex, which is adjacent to a $b$-vertex of every color except for its own.
We remark that $b^*$-vertices are also called \emph{nice} vertices in the literature~\cite{ZAKER2025370}.
It is easy to see that every $b^*$-vertex is a $b$-vertex, and every $b^*$-coloring is a $b$-coloring.
The motivation for this definition is another heuristic, which we call the $b^*$-heuristic, and it works as follows. Given a coloring that is \emph{not} a $b^*$-coloring, we can eliminate an arbitrary color class $C$ as follows.
We iterate over all vertices $u \in C$; if $u$ is not a $b$-vertex, then we can recolor it to some color not present in its neighborhood.
On the other hand, if $u$ is a $b$-vertex (but not a $b^*$-vertex), then there is a color class $C'$ such that no neighbor $v$ of $u$ in $C'$ is a $b$-vertex, so we first recolor every such vertex $v$ to some color not present in its neighborhood, and then we recolor $u$ to $C'$.
After processing all such vertices $u$, we have eliminated the color class $C$.
Hence, the $b^*$-heuristic is stronger than the $b$-heuristic in the sense that it improves even some $b$-colorings.
The \emph{$b^*$-chromatic number} $b^*(G)$ is defined analogously to the $b$-chromatic number as the largest integer $k$ such that there is a $b^*$-coloring of $G$ with $k$ colors, and it describes the worst-case behavior of the $b^*$-heuristic.

In some sense, $b^*$-colorings are simpler than general $b$-colorings because of their locality: it is easy to see that if $G$ is a graph with connected components $G_1, \ldots, G_m$, then $b^*(G) = \max_{1 \le i \le m} b^*(G_i)$~\cite[Proposition 3]{ZAKER2025370}, which is not true for the $b$-chromatic number\footnote{A star has $b$-chromatic number 2, but a disjoint union of stars may have arbitrarily high $b$-chromatic number.}.
However, this locality does not yield a polytime algorithm for \bscoloring because there is a simple reduction from \bcoloring~\cite{ZAKER2025370}, see Proposition~\ref{prop:b-to-bs}.
We remark that studying $b^*$-colorings may further our understanding of the classical $b$-colorings: for example, Zaker~\cite{zaker2026improved} obtained a sharp upper bound on the $b$-chromatic number in terms of the chromatic number and the independence number using $b^*$-colorings.

We say that a graph $G$ is \emph{$b^*$-monotonic}~\cite{ZAKER2025370} if $b^*(H_1) \le b^*(H_2)$ whenever $H_1$ is an induced subgraph of $H_2$ and $H_2$ is an induced subgraph of $G$.
Perhaps surprisingly, not all graphs are $b^*$-monotonic\footnote{\label{fn:mono}If $G$ (resp. $H$) is the graph obtained from $K_{n+1, n+1}$ (resp. $K_{n, n}$) by deleting a matching of size $n$, then $H$ is an induced subgraph of $G$ but $b^*(H) = n$ and $b^*(G) = 2$~\cite{ZAKER2025370}.}.
However, it is known that block graphs~\cite[Proposition 6]{ZAKER2025370} and cactus graphs~\cite[Proposition 8]{ZAKER2025370} are $b^*$-monotonic.
An interesting question is relating $b^*$-monotonicity to girth.
\begin{question}[{\cite[Problem 4]{ZAKER2025370}}]\label{question:girth}
What is the smallest integer $g$ such that all graphs of girth at least $g$ are $b^*$-monotonic?
\end{question}
It is easy to see that the integer $g$ in Question~\ref{question:girth} is at least 5, see Footnote~\ref{fn:mono}.

\subsection*{Contribution 1: Large girth graphs}

In Section~\ref{sec:girth}, we study $b^*$-colorings of large-girth graphs. First, we prove that the integer $g$ in Question~\ref{question:girth} is at most 7, which means that $g \in \{5, 6, 7\}$, see Theorem~\ref{thm:girth7}. In fact, we prove a stronger statement.
There is a trivial upper bound on $b^*(G)$, namely there has to be a vertex with many neighbors of high enough degree (because only such vertices may become $b^*$-vertices).
Formally, the $m^*$-degree of a vertex $u$ is the maximum integer $k$ such that $u$ has at least $k$ neighbors of degree at least $k$, and the $m^*$-degree of a graph $G$, denoted $m^*(G)$, is the maximum $m^*$-degree over all vertices of $G$.
Theorem~\ref{thm:girth7} proves not only that graphs $G$ of girth at least 7 are $b^*$-monotonic, but also that $b^*(G) = m^*(G) + 1$.

Second, we study regular graphs of high girth. It is easy to see that if $G$ is a $d$-regular graph, then $b^*(G) \le b(G) \le d+1$.
A question studied in the literature asks when the latter inequality is tight; it is known that girth at least 6 suffices.

\begin{theorem}[\cite{kouider2004b,elsahili2009about}]\label{thm:bg6}
If $G$ is a $d$-regular graph of girth at least 6, then $b(G) = d + 1$.
\end{theorem}

We observe that Theorem~\ref{thm:bg6} holds even if we replace $b(G)$ with $b^*(G)$, see Proposition~\ref{prop:bsrg6}. Hence, we may focus on graphs of girth 5.
It has been proven that small $d$ suffices for $b(G)$, unless $G$ is the famous Petersen graph.

\begin{theorem}[{\cite[Theorem 5]{BLIDIA20091787}}]\label{thm:bPeter}
If $G$ is a $d$-regular graph such that $\girth(G) \ge 5$, $G$ is not the Petersen graph, and $d \le 6$, then $b(G) = d+1$.
\end{theorem}

We generalize this theorem to $b^*$-colorings with a stronger bound on $d$. Moreover, we add a connectivity requirement because, for example, the disjoint union of two Petersen graphs would also be a counterexample.

\newcommand{\thmPG}{If $G$ is a connected $d$-regular graph such that $\girth(G) \ge 5$, $G$ is not the Petersen graph, and $d \le 3$, then $b^*(G) = d+1$.}

\begin{theorem}\label{thm:bsPeter}
\thmPG
\end{theorem}

A natural question is whether Theorem~\ref{thm:bPeter} (resp. our Theorem~\ref{thm:bsPeter}) holds even for graphs of large degree.
This question has been stated as the following conjecture.

\begin{conjecture}[{\cite{BLIDIA20091787}}]\label{conj:bd5}
Every $d$-regular graph with girth at least $5$, different from the Petersen graph, has a $b$-coloring with $d + 1$ colors.
\end{conjecture}

Recently, Dettlaff, Furmańczyk, Peterin, Roux, and Ziemann~\cite{DETTLAFF2024128914} presented two approaches to Conjecture~\ref{conj:bd5}. One of these approaches yields the following theorem (in the statement, $G[N^2[u]]$ is the subgraph of $G$ induced on the vertices at distance at most 2 from $u$).

\begin{theorem}[{\cite[Theorem 3.3]{DETTLAFF2024128914}}]\label{thm:DFPRZ}
Let $G$ be a $d$-regular graph, $d \ge 7$, with $\girth(G) = 5$. If there exists a vertex $u$ in $G$ contained in at most $5$ cycles of length 6 in $G[N^2[u]]$, then $b(G) = d+1$.
\end{theorem}

Our next contribution is Theorem~\ref{thm:d2}, which is similar to Theorem~\ref{thm:DFPRZ}.
However, our theorem is stronger since we allow the number of cycles containing $u$ to depend on $d$.
Hence, for $d \ge 11$, we obtain a strengthening of Theorem~\ref{thm:DFPRZ}, see Corollary~\ref{cor:d2}.
This strengthening from a constant $5$ to a function of $d$ was explicitly stated as an open problem~\cite[Section 5]{DETTLAFF2024128914}.

\newcommand{\thmd}{Let $G$ be a connected $d$-regular graph of girth at least $5$ such that $d \ge 4$.
If there is a vertex $u$ in $G$ contained in at most $\lceil \frac{d}{2} \rceil$ cycles of length 6 in $G[N^2[u]]$, then $b^*(G)=d+1$.}

\begin{theorem}\label{thm:d2}
\thmd
\end{theorem}

\begin{corollary}\label{cor:d2}
If $G$ is a $d$-regular graph of girth at least $5$ such that $d \ge 11$ and there is a vertex $u$ in $G$ contained in at most $\lceil \frac{d}{2} \rceil > 5$ cycles of length 6 in $G[N^2[u]]$, then $b(G)=d+1$.
\end{corollary}

\subsection*{Contribution 2: Parameterized complexity}

The aforementioned \XP algorithm for \bcoloring parameterized by clique-width by Jaffke, Lima, and Lokshtanov~\cite{JaffkeLL24} unified previously known polytime algorithms~\cite{IRVING1999127,BonomoSSV15,CamposSMS09} and initiated the study of structural parameterizations of this problem.
They~\cite{JaffkeLL24} also proved that \bcoloring is \FPT when parameterized by the vertex cover number.
Subsequently, Jaffke, Lima, and Sharma~\cite{JaffkeLS23} proved that \bcoloring is \XNLP-complete when parameterized by path-width, \FPT when parameterized by neighborhood diversity or the twin cover number, and \paraNP-complete when parameterized by twin-width or mim-width.
Finally, Balabán~\cite{DBLP:conf/mfcs/Balaban26} proved that \bcoloring is \Wone-hard when parameterized by tree-depth and \FPT when parameterized by the feedback edge number, by distance to co-cluster or by some other parameters, see Figure~\ref{fig:hierarchy}.

\begin{figure}[t]
\scalebox{0.75}{
\begin{tikzpicture}[every node/.style={draw, rectangle}]
  \node[fill=red!30] (twin) at (0,-2) {Twin-width~\cite{JaffkeLS23}};
  \node[fill=red!30] (mim) at (-4,-2) {Mim-width~\cite{JaffkeLS23}};
  \node[fill=blue!30] (clique) at (-2,-3) {Clique-width~\cite{JaffkeLL24}};
  \node[fill=blue!30]  (tw) at (-2,-4) {Tree-width};
  \node[fill=blue!30] (pw) at (-2, -5) {Path-width~\cite{JaffkeLS23}};
  \node[fill=blue!30] (td) at (-2,-6) {Tree-depth~\cite{DBLP:conf/mfcs/Balaban26}};
  \node[fill=green!30] (vi) at (-2,-7) {Vertex integrity~\cite{DBLP:conf/mfcs/Balaban26}};
  \node [align=center, fill=green!30] (vc) at (-2,-8) {Vertex cover number~\cite{JaffkeLL24}};
  
  \node [align=center, fill=green!30] (nd) at (-9.3,-6.8) {Neighborhood \\diversity~\cite{JaffkeLS23}};
  \node [align=center, fill=green!30] (tc) at (-6.8,-6.8) {Twin cover\\ number~\cite{JaffkeLS23}};
   
  \node [align=center, fill=green!30] (fen) at (4.5,-7.8) {Feedback edge \\number~\cite{DBLP:conf/mfcs/Balaban26}};
  \node [align=center] (fvn) at (4.5,-5.8) {Feedback vertex \\number};
  
  \node [align=center, fill=green!30] (bw) at (1.5,-8) {Band-width~\cite{DBLP:conf/mfcs/Balaban26}};
  \node [align=center, fill=green!30] (cutw) at (1.5,-7) {Cut-width~\cite{DBLP:conf/mfcs/Balaban26}};
  \node [align=center, fill=green!30] (carving) at (1.5,-6) {Carving-width~\cite{DBLP:conf/mfcs/Balaban26}};
  
  \draw[{Latex[length=2mm]}-] (tw.south east) -- (carving.north);
  \draw[{Latex[length=2mm]}-] (carving.south) -- (cutw.north);
  \draw[{Latex[length=2mm]}-] (cutw.south) -- (bw.north);
  \draw[{Latex[length=2mm]}-] (pw.south east) -- (cutw.north west);
  
  \node (modw) at (-9.3, -5) {Modular width};
  
  \node[align=center] (dtc) at (-6.8, -4.8) {Distance \\to cluster};
  \node[align=center, fill=green!30] (dtcc) at (-4.5, -4.8) {Distance \\to co-cluster~\cite{DBLP:conf/mfcs/Balaban26}};

  \draw[{Latex[length=2mm]}-] (dtc.south) -- (tc.north);
  \draw[{Latex[length=2mm]}-] (clique.south west) -- (dtc.north);
  
  \draw[{Latex[length=2mm]}-] ([xshift=10pt]clique.south west) -- (dtcc.north);
  \draw[-{Latex[length=2mm]}] (vc.north west) -- (dtcc.south);

  \draw[{Latex[length=2mm]}-] (twin.south) -- (clique.north east);
  \draw[{Latex[length=2mm]}-] (mim.south) -- (clique.north west);
  
  \draw[{Latex[length=2mm]}-] (clique.south) -- (tw.north);
  \draw[{Latex[length=2mm]}-] (tw.south) -- (pw.north);
  \draw[{Latex[length=2mm]}-] (pw.south) -- (td.north);
  \draw[{Latex[length=2mm]}-] (td.south) -- (vi.north);
  \draw[{Latex[length=2mm]}-] (vi.south) -- (vc.north);
  
  \draw[{Latex[length=2mm]}-] (tw.east) -- (fvn.north);
  \draw[{Latex[length=2mm]}-] (fvn.south) -- (fen.north);
  
  \draw[{Latex[length=2mm]}-] (modw.south) -- (nd.north);
  \draw[{Latex[length=2mm]}-] (modw.south east) -- (tc.north);
 
  \draw[{Latex[length=2mm]}-] (nd.south east) -- (vc.west);
  \draw[{Latex[length=2mm]}-] (tc.south east) -- (vc.north west);
  
  \draw[{Latex[length=2mm]}-] (clique.west) -- (modw.north);
\end{tikzpicture}}

\caption{Parameterized complexity of \bcoloring under various structural parameters (replicated from~\cite{DBLP:conf/mfcs/Balaban26} with permission). A directed path from a parameter $\alpha$ to a parameter $\beta$ indicates that $\beta \le f(\alpha)$ for some computable function $f$. Green stands for \FPT, blue is \textsf{W[1]}-hard and \textsf{XP}, red is \textsf{paraNP}-complete, and white means that it is unknown whether the problem is \FPT or \textsf{W[1]}-hard under that parameterization. We prove that the complexity landscape for \bscoloring looks exactly the same.}
\label{fig:hierarchy}
\end{figure}
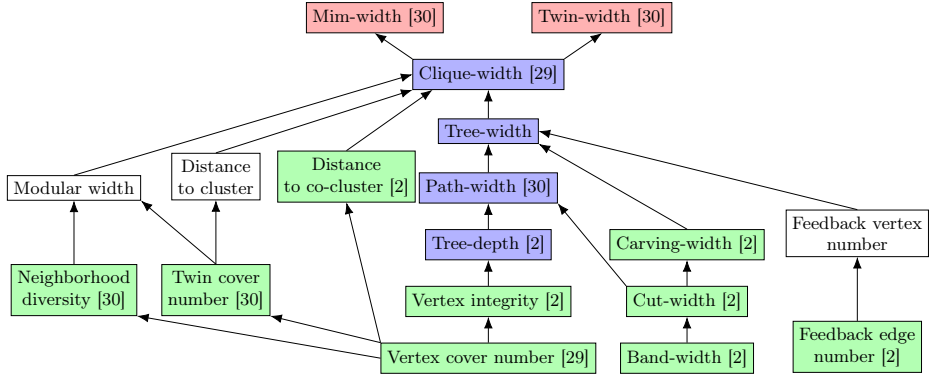

In Section~\ref{sec:par-complex}, we prove that all these results about the parameterized complexity of \bcoloring are true also for \bscoloring.
For the hardness results, it suffices to observe that the reduction showing \NP-hardness of \bscoloring (Proposition~\ref{prop:b-to-bs}) works also in our parameterized settings.
In addition, we observe that \bscoloring is \Wone-hard parameterized by the number of colors $k$, by reducing from the analogous result about \bcoloring~\cite{DBLP:journals/jcss/PanolanPS17}.

For vertex integrity and carving-width (and hence also for the parameters below them in Figure~\ref{fig:hierarchy}), the \FPT algorithms for \bcoloring work as follows.
If the number of colors $k$ is large compared to the parameter $p$, then we answer NO because there are not enough vertices which could become $b$-vertices.
Otherwise, $w + k$ is bounded by a function of $p$, where $w$ is the tree-width of the graph, and we may use dynamic programming.
We show that the same strategy works also for \bscoloring (we use the aforementioned concept of $m^*$-degree).

We transform the \XP algorithm for \bcoloring parameterized by clique-width~\cite{JaffkeLL24} into an algorithm for \bscoloring by first guessing a vertex $u$, which should become a $b^*$-vertex, and then ensuring that all $b$-vertices are in the closed neighborhood of $u$.
We remark that Zaker~\cite[Problem 3]{ZAKER2025370} asked for classes on which $b^*(G)$ can be determined in polynomial time.
Our result implies that any class of bounded clique-width admits such a polytime algorithm.

The \FPT algorithm for \bcoloring parameterized by neighborhood diversity~\cite{JaffkeLS23} works by solving a bounded number of Integer Linear Programming (ILP) instances. To transform it into an algorithm for \bscoloring, it suffices to restrict ourselves only to some of the ILP instances (those in which all $b$-vertices must be in the closed neighborhood of some vertex).
The \FPT algorithm for \bcoloring parameterized by the twin cover number~\cite{JaffkeLS23} works by applying a few reduction rules, which yields an equivalent instance of small neighborhood diversity.
We prove that the same reduction rules work also for \bscoloring.
To transform the \FPT algorithm parameterized by distance to co-cluster~\cite{DBLP:conf/mfcs/Balaban26}, it suffices to modify one definition, see Definition~\ref{def:candidates-and-flexible}.

Perhaps the most interesting parameterization is that by the feedback edge number.
The \FPT algorithm for \bcoloring using this parameterization~\cite{DBLP:conf/mfcs/Balaban26} is rather complex; it needs to handle several forbidden configurations, which are similar to the ``pivoted trees'' used in the polytime algorithm for \bchromatic on trees~\cite{IRVING1999127}.
However, our new \FPT algorithm for \bscoloring is much simpler: it uses one easy-to-prove lemma from~\cite{DBLP:conf/mfcs/Balaban26} but otherwise it is independent of the algorithm for \bcoloring.
The reason why our simpler algorithm exists is the locality of $b^*$-colorings: it ensures that the configurations considered in~\cite{DBLP:conf/mfcs/Balaban26} cannot occur.

\section{Preliminaries}\label{sec:prelims}
For integers $i$ and $j$, we define $[i,j] = \{a \in \mathbb Z \sep i \le a \le j\}$ and $[i] = [1, i]$.
For a set $S$, we denote the power set of $S$ by $2^S$, and $\bigcup S := \{u \sep \exists a \in S\colon u\in a\}$.
Let $f\colon A \rightarrow B$ be a partial function.
The \emph{domain} of $f$, denoted $\dom(f)$, is the subset of $A$ on which $f$ is defined.
We say that $f$ is \emph{total} if $\dom(f) = A$.
When we say a function, we mean a total function.
The \emph{range} of $f$, denoted $\range(f)$, is the set $\{b \in B\sep \exists a\in A\colon f(a) = b\}$.
For $A'\seq A$, we denote the function $\{(a, b) \in f \sep a \in A'\}$ by $f \upharpoonright A'$.
For $a \in A$ and $b \in B$, we define $f[a \mapsto b]$ to be the function $f \setminus \{(a, f(a))\} \cup \{(a, b)\}$.

Let $G$ be a graph (we will consider only finite simple graphs).
We assume familiarity with basic concepts in graph theory~\cite{Diestel}.
Given a set of vertices $U\seq V(G)$, we will use $\overline{U}$ to denote the set $V(G) \setminus U$, $G[U]$ to denote the graph induced on $U$, and $G - U$ to denote the graph $G[\overline{U}]$.
Similarly, for an edge set $F \seq E(G)$, $G-F$ denotes the graph obtained from $G$ by removing the edges in $F$.
We use $H \seq_i G$ to say that $H$ is an induced subgraph of $G$.

The \emph{length} of a path or a cycle is the number of edges it contains.
The \emph{girth} of $G$, denoted $\girth(G)$, is the length of the shortest cycle in $G$ (or $\infty$ if $G$ is acyclic).
The \emph{distance} between two vertices $u$ and $v$, denoted $\dist_G(u, v)$,
is the length of the shortest path between them (or $\infty$ if no such path exists).
If $v \in V(G)$ and $U \seq V(G)$ is non-empty, then $\dist_G(v, U) = \min\{\dist_G(u, v)\sep u \in U\}$.
For $U \seq V(G)$ and $v \in V(G)$, a $v$-$U$ path is a path $P$ in $G$ with endpoints $v$ and $u$ such that $V(P)\cap U = \{u\}$.

For $u\in V(G)$, we define the open neighborhood $N_G(u) := \{v \in V(G)\sep uv \in E(G)\}$, the closed neighborhood $N_G[u] := N_G(u) \cup \{u\}$, and their distance-2 variants $N^2_G[u] = \{v \in V(G)\sep \dist_G(u,v) \le 2\}$ and $N^2_G(u) = N^2_G[u] \setminus \{u\}$.
Similarly, for $S\seq V(G)$, we define $N_G(S) = \bigcup_{u \in S} N_G(u)$ and $N_G[S] = N_G(S) \cup S$.
We say that $u, v \in V(G)$ are \emph{twins} if $N_G(u) \setminus \{v\} = N_G(v) \setminus \{u\}$.
When $G$ is clear from the context, we will omit the subscript and write, e.g., $\dist(u, v)$ or $N(u)$.

The \emph{$m$-degree} of $G$, denoted $m(G)$, is the maximum integer $k$ such that there are at least $k$ vertices of degree at least $k-1$.
For a vertex $u \in V(G)$, we define the \emph{$m^*$-degree} of $u$, denoted $m^*(u)$, to be the maximum integer $k$ such that $u$ has at least $k$ neighbors of degree at least $k$.
Furthermore, we define the $m^*$-degree of $G$ as $m^*(G) = \max_{u \in V(G)} m^*(u)$.
Observe that $m^*(G) \le m(G) -1$.

We will use the famous Hall's theorem. 

\begin{theorem}[Hall's theorem, \cite{hall1987representatives}]\label{thm:Hall}
Let $G$ be a bipartite graph and let $X$ be one of the two parts.
There is a matching in $G$ covering $X$ if and only if $|N_G(W)| \ge |W|$ for all $W \seq X$.
\end{theorem}

\subsection{Colorings}
A (partial) \emph{$k$-coloring of $G$} is a (partial) function $\chi\colon V(G)\rightarrow [k]$.
We say that $\chi$ is \emph{proper} if $\chi(u)\ne\chi(v)$ for each $uv\in E(G)$.
For the following definitions, let us fix a graph $G$ and a $k$-coloring $\chi$.

We say that $u\in V(G)$ is a \emph{$b$-vertex} in $\chi$ (or a $\chi$-$b$-vertex) if $\chi(N[u]) = [k]$, and that $\chi$ is a \emph{$k$-$b$-coloring} if it is proper and for each color $c \in [k]$, there is a $b$-vertex $u$ in $\chi$ such that $\chi(u) = c$.
A coloring is a \emph{$b$-coloring} if it is a $k$-$b$-coloring for some integer $k$.
The \emph{$b$-chromatic number} of $G$, denoted $b(G)$, is the maximum integer $k$ such that $G$ admits a (total) $k$-$b$-coloring.

We say that $u$ is a \emph{$b^*$-vertex} in $\chi$ (or a $\chi$-$b^*$-vertex) if for each color $c \in [k]\setminus \{\chi(u)\}$, there is a $b$-vertex $v \in N(u)$ such that $\chi(v) = c$.
A proper (partial) $k$-coloring which contains a $b^*$-vertex is called a (partial) \emph{$k$-$b^*$-coloring.}
A coloring is a \emph{$b^*$-coloring} if it is a $k$-$b^*$-coloring for some integer $k$.
The \emph{$b^*$-chromatic number} of $G$, denoted $b^*(G)$, is the maximum integer $k$ such that $G$ admits a (total) $k$-$b^*$-coloring.
We say that $G$ is \emph{$b^*$-monotonic} if $b^*(H_1) \le b^*(H_2)$ for all $H_1 \seq_i H_2 \seq_i G$.
In some of our algorithms, we will need the following definition.

\begin{definition}[{\cite[Definition 2.3]{JaffkeLL24}}]
A \emph{potential $k$-$b$-coloring} of $G$ is a pair $(\chi, B)$, where $\chi$ is a proper $k$-coloring of $G$ and $B \seq V(G)$ is a set such that for every color $c \in [k]$, we have $|\chi^{-1}(c) \cap B| \le 1$.
\end{definition}

Informally, $B$ is a set of vertices that we want to become $b$-vertices in a coloring $\psi \supseteq \chi$ of some supergraph of $G$.
We remark that in~\cite{JaffkeLL24}, potential $k$-$b$-colorings are called \emph{partial $b$-colorings}; we chose to rename them to avoid confusion with our partial $b^*$-coloring (which is a partial coloring containing a $b^*$-vertex).

\subsection{Graph parameters}\label{sub:gp}

Let us fix a graph $G$.
We will not need the definition of clique-width.
Instead, we will use \emph{module-width}~\cite[defined under the name modular-width]{rao2008clique}, which is a parameter equivalent to clique-width (a graph class has bounded clique-width if and only if it has bounded module-width).
We denote the set of leaves of a tree $T$ by $L(T)$.
A \emph{rooted branch decomposition} of $G$ is a pair $(T, \ca L)$ of a rooted tree $T$ of maximum degree at most 3 and a bijection $\ca L \colon V(G) \to L(T)$.
Let us fix $t \in V(T)$. We denote by $T_t$ the subtree of $T$ rooted at $t$, and we define $V_t = \{v \in V(G) \sep \ca L(v) \in L(T_t)\}$ and $G_t = G[V_t]$.
We define an equivalence relation $\sim_t$ on $V_t$ as follows: $\forall u, v \in V_t \colon u \sim_t v \Leftrightarrow N_G(u) \cap \overline{V_t} = N_G(v) \cap \overline{V_t}$.
The \emph{module-width} of $(T, \ca L)$ is $\max_{t \in V(T)} |V_t/{\sim_t}|$ and the \emph{module-width of $G$} is the minimum module width over all rooted branch decompositions of $G$.

Now we define the remaining structural parameters.
An edge set $F \seq E(G)$ is called a \emph{feedback edge set} if $G-F$ is acyclic, and the \emph{feedback edge number} of $G$ is the size of a minimum feedback edge set in $G$.

An \emph{ND-partition} $\ca P$ is a partition of $V(G)$ such that each part of $\ca P$ is a clique or an independent set, and for all distinct parts $P_1, P_2 \in \ca P$, either $uv \in E(G)$ for all $u \in P_1, v \in P_2$ or $uv \notin E(G)$ for all $u \in P_1, v \in P_2$.
The \emph{neighborhood diversity} of $G$ is the minimum number of parts in any ND-partition of $G$.

A set $S \seq V(G)$ is a \emph{twin cover}, if for each edge $uv \in V(G)$, either $\{u, v\} \cap S \ne \emptyset$, or $u$ and $v$ are twins in $G$. The \emph{twin cover number} of
$G$ is the smallest size of any twin cover of $G$.

Given a graph class $\ca G$, we say that $S\seq V(G)$ is a \emph{$\ca G$-modulator of $G$} if $G-S \in \ca G$. An integer $p$ is the \emph{distance to $\ca G$ of $G$} if $p$ is the size of a minimum $\ca G$-modulator of $G$.
A \emph{cluster graph} is a disjoint union of cliques and a \emph{co-cluster graph} is a complement of a cluster graph, i.e., a complete multipartite graph.

Finally, we define tree-width.
A \emph{tree decomposition} of $G$ is a pair $(T, \ca B = \{B_t \mid t \in V(T)\})$, where $T$ is a tree and $B_t \seq V(G)$ for every $t \in V(T)$, satisfying the following conditions.
\begin{compactenum}
\item $\bigcup_{t \in V(T)} B_t = V(G)$.
\item For each $uv \in E(G)$, there is some $t \in V(T)$ such that $\{u, v\} \subseteq B_t$.
\item For each $v \in V(G)$, $T[\{t \in V(T) \sep v \in B_t\}]$ is connected.
\end{compactenum}
The \emph{width} of a tree decomposition is $\max_{t \in V(T)} |B_t| - 1$ and the \emph{tree-width} of $G$ is the minimum width over all its tree decompositions.
For a definition of a nice tree decomposition, see~\cite{CyganFKLMPPS15}.

\subsection{Parameterized complexity}

We assume familiarity with basics of parameterized complexity~\cite{flum2006parameterized}, namely with the classes \FPT, \XP, \Wone, and \paraNP.

The class \XNLP is defined as the class of those parameterized problems that can be solved by a non-deterministic algorithm that simultaneously uses time at most $f(p)\cdot n^{\ca O(1)}$ and space at most $f(p)\cdot \log n$, where $n$ is the length of the input, $p$ is the parameter, and $f$ is a computable function; see~\cite{bodlaender2024parameterized,elberfeld2015space} for more information on this class.

\section{Graphs of large girth}\label{sec:girth}

First, we prove a simple lemma, which gives us a partial $b^*$-coloring.

\begin{lemma}\label{lem:g6-init}
If $G$ is a graph of girth at least 6, then there is a partial $b^*$-coloring $\chi$ of $G$ using $k := m^*(G)+1$ colors such that $|\dom(\chi)| = k^2 -2k + 2$.
\end{lemma}
\begin{proof}
By definition of $k$, there is a vertex $u \in V(G)$ and distinct vertices $u_1, \ldots, u_{k-1} \in N(u)$ such that each of them has degree at least $k-1$.
For each $i \in [k-1]$, we define $\chi(u_i) = i$, and for each $j \in [k-1] \setminus \{i\}$, we define $\chi(v) = j$ for a distinct, previously uncolored vertex $v \in N(u_i) \setminus \{u\}$.
Note that such a vertex $v$ can always be found since $N(u_i) \cap N(u_j) = \{u\}$ for all $1 \le i < j \le k$ (which holds since $G$ has high enough girth).
Finally, we define $\chi(u) = k$.

Clearly, $u$ is a $b^*$-vertex and $|\dom(\chi)| = k^2 -2k + 2$.
Since $G$ has girth at least 6, we know that the subgraph of $G$ induced by the vertices at distance at most 2 from $u$ is a tree.
Hence, $\chi$ is proper, which concludes the proof.
\qed\end{proof}

We will need the following definition: a pair $(u, U)$ is called a \emph{$\chi$-$b^*$-core} if $u \in V(G)$ and $U \seq V(G)$ is a minimal set such that $u$ is a $b^*$-vertex in $\chi \upharpoonright U$.
Observe that $u \in U$ and that $\dist(u, u') \le 2$ for every $u' \in U$.
Informally, the following lemma says that vertices far from a $b^*$-vertex $u$ can be safely recolored without causing $u$ to stop being a $b^*$-vertex.

\begin{lemma}\label{lem:recolor}
Let $G$ be a graph, let $\chi$ be a partial $k$-coloring of $G$, let $(u, U)$ be a $\chi$-$b^*$-core, and let $v \in \dom(\chi) \setminus U$.
If $v$ is not a $b$-vertex in $\chi$, 
then there is a color $c \in [k] \setminus \{\chi(v)\}$ such that $u$ is a $b^*$-vertex in $\chi' := \chi[v \mapsto c]$ and $\chi'$ is proper.
\end{lemma}
\begin{proof}
Since $v$ is not a $b$-vertex, there is a color $c \in [k] \setminus \chi(N[v])$, and we define $\chi' = \chi[v \mapsto c]$.
Clearly, $\chi'$ is proper.
Let $w \in U$ be a $b$-vertex in $\chi \upharpoonright U$.
Since $v \notin U$, we have $\chi \upharpoonright U = \chi' \upharpoonright U$, which implies that $w$ is a $b$-vertex also in $\chi'$.
Hence, $u$ is a $b^*$-vertex also in $\chi'$, which concludes the proof.
\qed\end{proof}

Now we are ready to prove one of our main results.

\begin{theorem}\label{thm:girth7}
If $G$ is a graph of girth at least $7$, then $b^*(G) = m^*(G) + 1$ and $G$ is $b^*$-monotonic.
\end{theorem}
\begin{proof}
Let $G$ be an $n$-vertex graph, let $k = m^*(G) + 1$, and let $i_0 = k^2 - 2k + 2$.
For each $i \in [i_0, n]$, we will construct a partial coloring satisfying the following invariant: $\chi_i$ is a $k$-$b^*$-coloring and $|\dom(\chi_i)| = i$. 
Then, the total coloring $\chi_n$ will certify that $b^*(G) = m^*(G) + 1$.
Let $\chi_{i_0}$ be the coloring provided by Lemma~\ref{lem:g6-init} and observe that it satisfies the invariant.

Suppose that $\chi_i$ is defined for some $i \in [i_0, n-1]$ and let $v \in \overline{\dom(\chi_i)}$.
For each $c \in [k]$, let $V_c = N(v) \cap \chi_i^{-1}(c)$.
First, suppose that for each $c \in [k]$, there is a $\chi_i$-$b$-vertex $w \in V_c$.
Observe that $w$ has degree at least $k$; it has $k-1$ neighbors colored by $\chi_i$ plus $v$.
Hence, $m^*(v) \ge k$, which is a contradiction with $m^*(G) = k-1$.
Hence, there is a color $c \in [k]$ such that $V_c$ contains no $\chi_i$-$b$-vertices; let us fix such a color $c$. Note that $V_c = \emptyset$ is possible.

Let $(u, U)$ be a $\chi_i$-$b^*$-core.
If $V_c \cap U = \emptyset$, then we successively apply Lemma~\ref{lem:recolor} to vertices in $V_c$ to obtain a proper coloring $\chi'$ such that $u$ is a $b^*$-vertex in $\chi'$ and $c \notin \chi'(N(v))$.
Note that after one application of the lemma, no other vertex in $V_c$ becomes a $b$-vertex because $V_c$ is an independent set in $G$, so the construction of $\chi'$ is valid.
Now the coloring $\chi_{i+1} := \chi' \cup \{(v, c)\}$ satisfies the invariant.
Hence, we may assume that $V_c \cap U \ne \emptyset$.

Suppose that there is a color $d \in [k] \setminus \{c\}$ such that $V_d$ contains no $\chi_i$-$b$-vertices.
Since $G[U]$ has diameter at most 4, $G$ has girth at least 7 and $V_c \cap U \ne \emptyset$, we have $V_d \cap U = \emptyset$.
Hence, we may, analogously to the previous paragraph, use Lemma~\ref{lem:recolor} to obtain a proper coloring $\chi'$ such that $u$ is a $b^*$-vertex in $\chi'$ and $d \notin \chi'(N(v))$, and the coloring $\chi_{i+1} = \chi' \cup \{(v, d)\}$ satisfies the invariant.

Finally, suppose that no such color $d$ exists, i.e., for each color $d \in [k] \setminus \{c\}$, there is a $\chi_i$-$b$-vertex $v_d \in V_d$.
Recall that by choice of $c$, $V_c$ contains no $\chi_i$-$b$-vertices.
Hence, for each $w \in V_c$, there is a color $c_w \in [k]$ such that $\chi_i[w \mapsto c_w]$ is proper.
Since $V_c$ is an independent set, the coloring $\chi' := \chi_i \setminus \{(w, c) \sep w \in V_c\} \cup \{(w, c_w) \sep w \in V_c\}$ is also proper.
Let us define $\chi_{i+1} = \chi' \cup \{(v, c)\}$, and observe that $\chi_{i+1}$ is proper.
Since $N(v)$ is an independent set in $G$, $v_d$ is a $b$-vertex colored with $d$ also in $\chi_{i+1}$, for every color $d \in [k] \setminus \{c\}$.
Hence, $v$ is a $b^*$-vertex in $\chi_{i+1}$, which means that $\chi_{i+1}$ satisfies the invariant.

We have proven that $b^*(G) = m^*(G) + 1$.
To show that $G$ is $b^*$-monotonic, it suffices to observe that $m^*(H_1) \le m^*(H_2)$ and $\girth(H_1) \le \girth(H_2)$ whenever $H_1$ and $H_2$ are two graphs satisfying $H_1 \seq_i H_2$.
\qed\end{proof}

An open question (see Question~\ref{question:girth}) is whether Theorem~\ref{thm:girth7} can be generalized to graphs of girth 6 or even 5 (for girth 4, it is false, see Footnote~\ref{fn:mono}).
For girth 6, the same strategy, i.e., showing that $b^*(G) = m^*(G) + 1$, might work.
However, for girth 5, we have a counterexample.

\begin{proposition}
There is a graph $G$ with girth 5 such that $b^*(G) < m^*(G) + 1$.
\end{proposition}
\begin{proof}
Let $G$ be the well-known Petersen graph.
By \cite[Theorem 1]{BLIDIA20091787}, $G$ admits no $4$-$b$-coloring, and hence also no $4$-$b^*$-coloring.
It can be easily seen that $m^*(G) = 3$.
Hence, $G$ is the desired graph.
\qed\end{proof}

\subsection{Regular graphs}

In the remainder of Section~\ref{sec:girth}, we focus on regular graphs.
First, we observe that in this setting, a partial $b^*$-coloring can be greedily extended into a total $b^*$-coloring.

\begin{observation}\label{obs:finish-regular}
If $G$ is a $d$-regular graph and there is a partial $b^*$-coloring $\chi$ with $d+1$ colors, then $b^*(G) = d +1$.
\end{observation}
\begin{proof}
Clearly, we have $b^*(G) \le d+1$.
Since we can greedily extend $\chi$ into a proper $(d+1)$-$b^*$-coloring of $G$, we also have $b^*(G) \ge d+1$.
\qed\end{proof}

Combining Lemma~\ref{lem:g6-init} and Observation~\ref{obs:finish-regular}, we immediately obtain the $b^*$ analogue of Theorem~\ref{thm:bg6}.

\begin{proposition}\label{prop:bsrg6}
If $G$ is a $d$-regular graph of girth at least 6, then $b^*(G) = m^*(G) +1 = d + 1$.
\end{proposition}
\begin{proof}
Let $\chi$ be the coloring provided by Lemma~\ref{lem:g6-init}.
Since $m^*(G) = d$, it suffices to apply Observation~\ref{obs:finish-regular} to $\chi$.
\qed\end{proof}

\newcounter{savedtheorem}
\setcounter{savedtheorem}{\value{theorem}}
\setcounter{theorem}{3}

Now we prove a variant of Theorem~\ref{thm:bPeter} for $b^*$-colorings.

\begin{theorem}
\thmPG
\end{theorem}

\setcounter{theorem}{\value{savedtheorem}}

\begin{proof}
If $d = 1$, then $G$ is a clique on two vertices, and the unique proper $2$-coloring of $G$ is a $b^*$-coloring.
If $d = 2$, then $G$ is a cycle of length at least 5.
Let $u_1u_2u_3u_4u_5$ be a path in $G$, and let $\chi = \{(u_1, 1), (u_2, 2), (u_3, 3), (u_4, 1), (u_5, 2)\}$ be a partial coloring of $G$.
Observe that $u_3$ is a $\chi$-$b^*$-vertex and that $\chi$ is proper (since $\chi(u_1) \ne \chi(u_5)$).
Hence, $b^*(G) = d+1$ by Observation~\ref{obs:finish-regular}.

Now suppose that $d = 3$.
Let $u \in V(G)$ be an arbitrary vertex, let $N(u) = \{u_1, u_2, u_3\}$, and for $i \in [3]$, let $\chi(u_i) = i$, $N(u_i) = \{u, v_i, w_i\}$, and $U_i = N(u_i) \setminus \{u\}$.
Moreover, let $\chi(u) = 4$, and let $U = U_1 \cup U_2 \cup U_3$.
Since $G$ has girth at least 5, a vertex $v \in U_i$, $i \in [3]$, has at most one neighbor in $U_j$, $j \in [3]$, and since $G$ is not the Petersen graph, $G[U]$ has at most 5 edges.
We proceed by induction on $5 - |E(G[U])|$.

First, suppose that $G[U]$ has exactly 5 edges.
If $v_i$ and $w_i$ had degree 1 in $G[U]$ for some $i \in [3]$, then there would have to be three edges in $G[U \setminus U_i]$, which is impossible.
Hence, without loss of generality, let $v_1$ and $v_2$ be the two vertices of degree 1 in $G[U]$ (and all other vertices of $U$ have degree 2 in $G[U]$).
Observe that $w_1$ (resp. $w_2$) has one neighbor in $U_2$ (resp. $U_1$) and one in $U_3$, and $v_3$ and $w_3$ have each one neighbor in $U_1$ and one in $U_2$.
Now it is routine to verify that $v_1$ must have a neighbor in $U_3$, say $v_3$, and that $E(G[U]) = \{v_1v_3, w_1w_3, w_2v_3, v_2w_3, w_1w_2\}$ is the only possibility which does not produce a cycle of length at most 4, see Figure~\ref{fig:notPetersen}.
Hence, we may define $\chi = \{(v_1, 3), (w_1, 2), (v_2, 3), (w_2, 1), (v_3, 2), (w_3, 1)\}$.
Clearly, $\chi$ is a partial $4$-$b^*$-coloring, and we conclude by Observation~\ref{obs:finish-regular}.

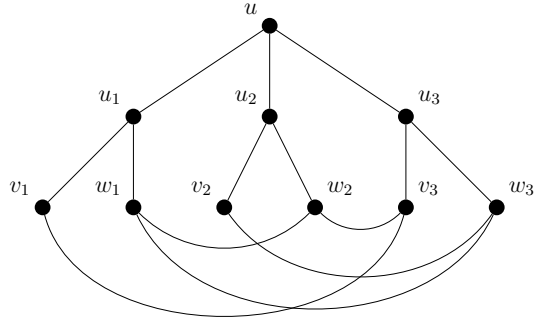
\begin{figure}[t]
\centering
\scalebox{0.75}{
\begin{tikzpicture}[
    scale=0.8,
    every node/.style={font=\large},
    dot/.style={circle, fill=black, inner sep=2.8pt, outer sep=0pt}
]

  \node[dot, label=above left:{$u$}] (w) at (0, 4) {};

  \node[dot, label=above left:{$u_1$}]  (w1) at (-3, 2) {};
  \node[dot, label=above left:{$u_2$}] (w2) at (0, 2) {};
  \node[dot, label=above right:{$u_3$}] (w3) at (3, 2) {};

  \node[dot, label=above left:{$v_1$}]  (v1) at (-5, 0) {};
  \node[dot, label=above left:{$w_1$}]  (u1) at (-3, 0) {};
  \node[dot, label=above left:{$v_2$}]  (u2) at (-1, 0) {};
  \node[dot, label=above right:{$w_2$}] (v2) at (1, 0) {};
  \node[dot, label=above right:{$v_3$}] (u3) at (3, 0) {};
  \node[dot, label=above right:{$w_3$}] (v3) at (5, 0) {};

  \draw (w) -- (w1);
  \draw (w) -- (w2);
  \draw (w) -- (w3);

  \draw (w1) -- (v1);
  \draw (w1) -- (u1);
  \draw (w2) -- (u2);
  \draw (w2) -- (v2);
  \draw (w3) -- (u3);
  \draw (w3) -- (v3);

  \draw (u1) to[bend right=45] (v2);
  \draw (v2) to[bend right=45] (u3);

  \draw (u2) to[bend right=55] (v3);
  \draw (u1) to[bend right=65] (v3);
  \draw (v1) to[bend right=75] (u3);

\end{tikzpicture}
}
\vspace*{-10pt}
\caption{An illustration of the proof of Theorem~\ref{thm:bsPeter}.}
\label{fig:notPetersen}
\end{figure}

Now suppose that $|E(G[U])| = m < 5$.
Intuitively, adding an edge can make coloring $G[U]$ only harder, so it suffices to show that an edge can always be added to $G[U]$ without violating the girth constraint.
Without loss of generality, assume that $v_1$ has degree at most 1 in $G[U]$ and that it has no neighbor in $U_2$.
Since at most one vertex in $U_2$ is adjacent to $w_1$, we may add an edge between $v_1$ and $U_2$, which concludes the proof.
\qed\end{proof}

\setcounter{savedtheorem}{\value{theorem}}
\setcounter{theorem}{6}

Finally, we prove a variant of Theorem~\ref{thm:DFPRZ} for $b^*$-colorings.
Recall that Theorem~\ref{thm:d2} implies a strengthening of Theorem~\ref{thm:DFPRZ}, see Corollary~\ref{cor:d2}.

\begin{theorem}
\thmd
\end{theorem}
\begin{proof}
Let $G$, $d$, and $u$ be as in the statement.
Let $N(u) = \{u_1, \ldots, u_d\}$ and for each $i \in [d]$, let $N(u_i) = \{u, u_i^1, \ldots, u_i^{d-1}\}$, $U_i = N(u_i) \setminus \{u\}$, and $\Delta_i = (\delta_i^1, \ldots, \delta_i^{d-1})$, where $\delta_i^j$ is the degree of $u_i^j$ in $G[N^2(u)]$, for every $j \in [d-1]$.
We may assume that $\Delta_i$ is non-increasing for every $i \in [d]$, and that $\Delta_i$ is (non-strictly) lexicographically greater than $\Delta_j$ for all $1 \le i \le j \le d$, i.e., informally, neighbors of $u_1$ have the largest degrees.
Let us define $\chi_0(u_i) = i$ for $i \in [d]$, and $\chi_0(u) = d+1$.
For each $i \in [d]$, we will define a coloring $\chi_i$ satisfying the following invariant: $\chi_0 \seq \chi_i$ and for every $j \in [i]$, $u_j$ is a $\chi_i$-$b$-vertex.
Trivially, this invariant is satisfied also by $\chi_0$.

Let $i \in [d]$ and suppose that $\chi_{i-1}$ satisfies the invariant.
Let $H$ be a bipartite graph with parts $U_i$ and $C_i := [d] \setminus \{i\}$ such that $vc \in E(H)$ for $v \in U_i$ and $c \in C_i$ if and only if $c \notin \chi_{i-1}(N(v))$.
Observe that given a perfect matching $\mu$ of $H$, we could define $\chi_i = \chi_{i-1} \cup \{(v, c) \sep v$ is matched with $c$ in $\mu\}$.
Hence, it suffices to find $\mu$ using Hall's theorem, see Theorem~\ref{thm:Hall}.
Let $W \seq U_i$; it suffices to verify that $|N_H(W)| \ge |W|$.
Since $G$ has girth at least $5$, each vertex in $U_j$ for $j \in [i-1]$ has at most one neighbor in $U_i$, which implies that each color $c \in C_i$ is forbidden for at most $i-1$ vertices, i.e., $c$ has degree at least $d-1-(i-1) = d-i$ in $H$.
Analogously, a vertex $v \in U_i$ has at most $i-1$ forbidden colors, i.e., $v$ has degree at least $d-i$ in $H$.
Hence, Hall's condition is satisfied if $|W| \le d-i$.

Suppose that $|W| \ge d-i+1$. First, suppose that $i \le \lceil \frac{d}{2} \rceil$.
Observe that \[d > d-1 \ge 2\cdot  \left\lceil \frac{d}{2} \right\rceil - 2 \ge 2\cdot i -2,\] which implies $d-i+1 > i-1$.
Hence, no color can be forbidden for all vertices in $W$, which implies $|N_H(W)| = d-1 \ge |W|$ as desired.

Second, suppose that $i > \lceil \frac{d}{2} \rceil$.
Suppose for contradiction that some vertex $v_i \in U_i$ has degree at least 3 in $G[N^2(u)]$.
Since we ordered vertices from $N(u)$ based on the degrees of their neighbors, we know that there is a vertex $v_j \in U_j$ of degree at least 3 in $G[N^2(u)]$, for every $j \in [i]$.
Observe that for every $j \in [i]$, there is a cycle $C^j$ in $G[N^2[u]]$ such that $V(C^j) = \{v_j, x, u_k, u, u_\ell, y\}$, where $k, \ell  \in [d] \setminus \{j\}$, $k \ne \ell$, $x\in N(v_j) \cap U_k$ and $y \in N(v_j) \cap U_\ell$.
It is easy to see that $C^j \ne C^k$ for $1 \le j < k \le i$, which is a contradiction since there are at most $\lceil \frac{d}{2} \rceil < i$ such cycles by choice of $u$.

Hence, every vertex $v \in U_i$ has degree at most 2 in $G[N^2(u)]$.
Since one neighbor of $v$ in $G[N^2(u)]$ is $u_i$, we know that $v$ has degree at least $d-2$ in $H$.
Hence, Hall's condition is fulfilled whenever $|W| \le d-2$, so it suffices to consider the case $W = U_i$.
Let $c \in C_i$ be a color. Since $c \notin \chi_{i-1}(U_c)$, we know that $c$ is forbidden for at most $d-2$ vertices of $U_i$, i.e., $c$ has positive degree in $H$.
Hence, $N_H(W) = C_i$ as desired.
We have verified Hall's condition in all cases, so $\mu$ exists and $\chi_i$ can be defined.

To conclude, observe that $\chi_d$ is a partial $b^*$-coloring, which means that it suffices to apply Observation~\ref{obs:finish-regular} to obtain the desired $b^*$-coloring.
\qed\end{proof}

\setcounter{theorem}{\value{savedtheorem}}

\section{Parameterized complexity}\label{sec:par-complex}

In this section, we study the parameterized complexity of \bscoloring.

First, we recall the following relation between $b$-colorings and $b^*$-colorings. If $G$ is a graph, then $G \lor K_1$ denotes the graph obtained from $G$ by adding a universal vertex, i.e., a new vertex connected to all original vertices of $G$.

\begin{proposition}[{\cite[Proposition 4]{ZAKER2025370}}]\label{prop:b-to-bs}
For any graph $G$, there is a $b$-coloring of $G$ with $k$ colors if and only if there is a $b^*$-coloring of $G \lor K_1$ with $k+1$ colors.
In particular, $b(G) = b^*(G \lor K_1) - 1$.
\end{proposition}

We remark that \cite[Proposition 4]{ZAKER2025370} states only that $b(G) = b^*(G \lor K_1) - 1$ but the same proof works also for our more general statement.
Informally, the new universal vertex receives color $k+1$. 
Using Proposition~\ref{prop:b-to-bs}, we show the following hardness results.

\begin{proposition}
The \bscoloring problem, with input $(G, k)$, is:
\begin{enumerate}
\item \Wone-hard parameterized by $k$ or by the tree-depth (td) of $G$.
\item \XNLP-hard parameterized by the path-width (pw) of $G$.
\item \paraNP-hard parameterized by the mim-width (mw) of $G$ or by the twin-width (tww) of $G$.
\end{enumerate}
\end{proposition}
\begin{proof}
All hardness results use the same reduction from \bcoloring, namely $(G, k) \mapsto (G \lor K_1, k+1)$.
The \Wone-hardness of \bchromatic parameterized by $k$ is proven in~\cite{DBLP:journals/jcss/PanolanPS17}, which implies that \bcoloring
is \Wone-hard parameterized by $k$ as well.
The hardness (\Wone-hardness, \XNLP-hardness or \paraNP-hardness, respectively) of \bcoloring is proven for tree-depth in~\cite{DBLP:conf/mfcs/Balaban26} and for path-width, mim-width, and twin-width in~\cite{JaffkeLS23}. 
By Proposition~\ref{prop:b-to-bs}, the reduction is correct.
Moreover, observe that for every parameter $p \in \{k, \td, \pw, \mw, \tww\}$, we have $p(G \lor K_1) \le p(G) + 1$.
Hence, the construction provides a polytime and parameterized logspace reduction; the latter is needed to transfer \XNLP-hardness~\cite{elberfeld2015space}.
\qed\end{proof}

Now we present an \FPT algorithm parameterized by tree-width plus the number of colors.

\begin{proposition}\label{prop:ktw}
The \bscoloring problem, with input $(G, k)$, is \FPT parameterized by $k + w$, where $w$ is the tree-width of $G$, and it can be solved in time $2^{\ca O(k \cdot w)} \cdot n^2$, where $n = |V(G)|$.
\end{proposition}
\begin{proof}
It suffices to slightly modify the dynamic programming algorithm for \bcoloring~\cite[Proposition 2.1]{JaffkeLL24}.
Before we run the DP algorithm, we first guess a vertex $u \in V(G)$.
Second, we compute a nice tree decomposition $(T, \ca B = \{B_t \sep t \in V(T)\})$ of $G$ of width $2w + 1$, see~\cite{Korhonen21,Kloks94}.
Similarly as in~\cite{JaffkeLL24}, for every $t \in V(T)$, we compute a set of entries $\tab(t)$ such that $(\gamma, C, P, \sigma) \in \tab(t)$ if there is a potential $b$-coloring $(\chi, B)$ of $G_t$ such that $B \seq N[u]$ and:
\begin{compactenum}
\item $\gamma \colon B_t \to [k]$ is a proper coloring with $\gamma = \chi \upharpoonright B_t$.
\item $P = B_t \cap B$.
\item $\sigma \colon P \to 2^{[k]}$ is a map such that for each $p \in P$, $\sigma(p) = \chi(N(p))$.
\item $C = \chi(B)$ and each vertex in $B \setminus P$ is a $\chi$-$b$-vertex.
\end{compactenum}

These requirements are the same as in~\cite{JaffkeLL24}, except that we added the requirement that $B \seq N[u]$.
This addition clearly ensures that if the algorithm finds a $k$-$b$-coloring, then it is a $k$-$b^*$-coloring.
The running time increases by a multiplicative factor of $n$ because of guessing the vertex $u$.
\qed\end{proof}

Using Proposition~\ref{prop:ktw}, we obtain several \FPT algorithms.

\begin{proposition}
The \bscoloring problem, with input $(G, k)$, is \FPT parameterized by $p$ if $p$ is the vertex cover number, vertex integrity, band-width, cut-width or carving-width of $G$.
\end{proposition}
\begin{proof}
Let $p$ be one of the five parameters from the statement.
It has been shown that there is a computable function $f_p$ such that $\tw(G) \le f_p(p(G))$ and $m(G) \le f_p(p(G))$, where $\tw$ is tree-width and $m$ is the $m$-degree, see~\cite{JaffkeLL24} for the vertex cover number, and~\cite{DBLP:conf/mfcs/Balaban26} for the other parameters.
If $k > f_p(p(G))$, then we answer NO because $b^*(G) \le m^*(G)+1 \le m(G)$ for every graph $G$ (recall that $G$ cannot have a $b^*$-coloring with more than $b^*(G)$ colors).
Otherwise, it suffices to use the algorithm parameterized by $\tw(G) + k$ from Proposition~\ref{prop:ktw}.
\qed\end{proof}

\subsection{Clique-width}

Let us fix an $n$-vertex graph $G$, an integer $k$, and a rooted branch decomposition $(T, \ca L)$ for $G$ of module-width $w$.
Recall that module-width and related terminology was defined in Section~\ref{sub:gp}.

First, we recall necessary definitions from~\cite{JaffkeLL24}.
Informally, we assign a $t$-type $(\phi, \xi)$ to every color $c \in [k]$ based on its behavior in a potential $k$-$b$-coloring $(\chi, B)$ of $G_t$: the bit $\xi$ says whether $c$ is the color of some vertex in $B$, and the function $\phi$ specifies, for each $Q \in V_t/{\sim_t}$, whether $Q$ contains a vertex colored with $c$ ($\ccontains$), whether some vertex in $B \cap Q$ misses $c$ in its neighborhood ($\cdemand$), or none of these two possibilities ($\cnone$).

\begin{definition}[{\cite[Definition 3.1, typo fixed]{JaffkeLL24}}]
If $t \in V(T)$ is a node, then a \emph{$t$-type} is a pair $(\phi, \xi)$ of a map $\phi \colon V_t/{\sim_t} \to \{\cnone,\ccontains,\cdemand\}$ and a bit $\xi \in \{0, 1\}$. We denote the set of all $t$-types by $\ctypes_t$.
\end{definition}

\begin{definition}[{\cite[Definition 3.2]{JaffkeLL24}}]\label{def:cw-has-type}
Let $t \in V(T)$, let $(\chi, B)$ be a potential $k$-$b$-coloring of $G_t$, let $c \in [k]$ be a color, and let $\tau = (\phi, \xi) \in \ctypes_t$ be a $t$-type. We say that \emph{$c$ has $t$-type $\tau$ in $(\chi, B)$} if, for $C = \chi^{-1}(c)$:
\begin{compactenum}
\item $\xi = |C \cap B|$ and
\item for each $Q \in V_t/{\sim_t}$,
\begin{compactenum}
\item if $Q \cap C \neq \emptyset$, and there is no $v \in (B \setminus C) \cap Q$ 	such that $N(v) \cap C = \emptyset$, then $\phi(Q) = \ccontains$,
\item if $Q \cap C = \emptyset$ and there exists some $v \in (B \setminus C) \cap Q$ such that $N(v) \cap C = \emptyset$, then $\phi(Q) = \cdemand$, and
\item if $Q \cap C = \emptyset$, and there is no $v \in (B \setminus C) \cap Q$ such that $N(v) \cap C = \emptyset$, then $\phi(Q) = \cnone$.
\end{compactenum}
\end{compactenum}
\end{definition}

Now we define a \emph{$t$-signature}, which says how many colors of each $t$-type should be present.

\begin{definition}[{\cite[Definitions 3.3 and 3.4]{JaffkeLL24}}]\label{def:cw-signature}
If $t \in V(T)$, then a \emph{$t$-signature} is a map $\csig_t \colon \ctypes_t \to \{0, 1, \ldots, k\}$ such that  $\sum_{\tau \in \ctypes_t}\csig_t(\tau) = k$.

If $(\chi, B)$ is a potential $k$-$b$-coloring of $G_t$, then $(\chi, B)$ \emph{has signature} $\csig_t$ if for each $t$-type $\tau \in \ctypes_t$, there are precisely $\csig_t(\tau)$ colors in $(\chi, B)$ that have $t$-type $\tau$ in $(\chi, B)$.
\end{definition}

The following lemma says that a set of $t$-signatures can be computed by combining a set of $s$-signatures and a set of $r$-signatures, where $s$ and $r$ are the children of $t$. 

\begin{lemma}[{\cite[implicit in the proof of Theorem 3.16]{JaffkeLL24}}]\label{lem:cw-black-box}
Let $t \in V(T)$ be a node with children $s, r \in V(T)$.
For $q \in \{s, r\}$, let $A_q$ be a set of $q$-signatures.
There is an algorithm, running in time $n^{2^{\ca O(w)}}$, that computes a set $A_t$ of $t$-signatures such that $\csig_t \in A_t$ if and only if there is a potential $b$-coloring $(\chi, B)$ of $G_t$ with signature $\csig_t$ and for all $q \in \{s, r\}$, there is a signature $\csig_q \in A_q$ such that $(\chi \upharpoonright V(G_q), B \cap V(G_q))$ has  signature $\csig_q$.
\end{lemma}

Now we are ready to modify the \XP algorithm for \bcoloring parameterized by module-width into an analogous algorithm for \bscoloring.

\begin{theorem}
The \bscoloring problem, with input $(G, k)$, is \XP parameterized by the clique-width $w$ of $G$, and it can be solved in time $n^{2^{\ca O(w)}}$.
\end{theorem}
\begin{proof}
Let $n = |V(G)|$ and $(T, \ca L)$ be a rooted branch decomposition for $G$ of module-width $w$.
First, we guess a vertex $u \in V(G)$: we will use Lemma~\ref{lem:cw-black-box} to decide whether there is a $k$-$b^*$-coloring in which $u$ is a $b^*$-vertex.
Our algorithm is a dynamic programming along $(T, \ca L)$, analogously to~\cite[Theorem 3.16]{JaffkeLL24}.
For every node $t \in V(T)$, we will define a set of $t$-signatures $\tab(t)$.

Let $t \in V(T)$ be a leaf of $T$ and let $v \in V(G)$ be such that $\ca L(v) = t$.
Clearly, the only equivalence class of $\sim_t$ is $\{v\}$.
For $\alpha \in \{\cnone,\ccontains,\cdemand\}$, let $\phi_\alpha$ be a function such that $\phi_\alpha(\{v\}) = \alpha$.
There are two possible signatures: either $v$ should not be in $B$ or it should.
The former case corresponds to the signature $\csig_t^1 = \{((\phi_\ccontains, 0), 1), ((\phi_\cnone, 0), k-1)\}$ and the latter case to the signature $\csig_t^2 = \{((\phi_\ccontains, 1), 1), ((\phi_\cdemand, 0), k-1)\}$.
In~\cite{JaffkeLL24}, $\tab(t)$ is defined as $\{\csig_t^1, \csig_t^2\}$.
However, we need to ensure that $B \seq N[u]$, so we define $\tab(t) = \{\csig_t^1, \csig_t^2\}$ if $v \in N[u]$, and $\tab(t) = \{\csig_t^1\}$ otherwise.

If $t \in V(T)$ is an internal node of $T$, then $\tab(t)$ is computed using Lemma~\ref{lem:cw-black-box}.
Let $r \in V(T)$ be the root of $T$ and let $\csig_r^* = \{((\{(V(G), \ccontains)\}, 1), k)\}$ be an $r$-signature.
The algorithm answers YES if $\csig_r^* \in \tab(r)$, and otherwise it answers NO.
Observe that there are $n$ choices for $u$, $2n-1$ nodes of $T$ and for each of them, the computation takes time at most $n^{2^{\ca O(w)}}$ by Lemma~\ref{lem:cw-black-box}, which means that the total running time is $n^{2^{\ca O(w)}}$, as desired.

What remains is to prove correctness.
First, suppose that the algorithm answers YES. By Lemma~\ref{lem:cw-black-box}, there is potential $b$-coloring $(\chi, B)$ with signature $\csig_r^*$.
By definition of $\csig_r^*$, $\chi$ is a $b$-coloring and each vertex $v \in B$ is a $\chi$-$b$-vertex.
Let $t \in V(T)$ be a leaf such that $\ca L(v) = t$, and let $P$ be the $r$-$t$ path in $T$.
By inductively applying Lemma~\ref{lem:cw-black-box} along $P$, there is, for every $q \in V(P)$, a $q$-signature $\csig_q \in \tab(q)$ such that $(\chi \upharpoonright V(G_q), B \cap V(G_q))$ has signature $\csig_q$.
Let $c = \chi(v)$ and observe that $|\chi^{-1}(c) \cap B \cap V(G_t)| = 1$.
Hence, by Definitions~\ref{def:cw-has-type} and~\ref{def:cw-signature}, $\csig_t^2 \in \tab(t)$, which implies $v \in N[u]$.
Therefore, $u$ is a $\chi$-$b^*$-vertex, as desired.

Conversely, suppose that there is a $k$-$b^*$-coloring $\chi$.
Let $u$ be a $\chi$-$b^*$-vertex and let $B \seq N[u]$ be such that for each color $c \in [k]$, there is a $\chi$-$b$-vertex in $B$.
We prove by ``leaves-to-root'' induction that for every $t \in V(T)$, there is a signature $\csig_t \in \tab(t)$ such that $(\chi \upharpoonright V(G_t), B \cap V(G_t))$ has signature $\csig_t$.
If $t \in V(T)$ is a leaf in $T$ and $\ca L^{-1}(t) = v$, then $\csig_t = \csig_t^1$ if $v \notin B$, and $\csig_t = \csig_t^2$ if $v \in B$.
If $t \in V(T)$ is an internal node of $T$, then $\csig_t$ exists by induction hypothesis and Lemma~\ref{lem:cw-black-box}.
Hence, $\csig_t \in \tab(t)$ for every $t \in V(T)$.
Finally, observe that $\csig_r = \csig_r^*$, which concludes the proof.
\qed\end{proof}

\subsection{Feedback edge number}

\newcommand{\red}{\textsf{red}}

Let us fix a graph $G$ and an integer $k \ge 4$.
Let $p_G$ be the feedback edge number of $G$.
We will need the following definition, which is illustrated in Figure~\ref{fig:fen-core}.

\begin{definition}[{\cite[Definition 17]{DBLP:conf/mfcs/Balaban26}}]\label{def:fen-core}
Let $S \seq V(G)$ be a set.
An \emph{$S$-outer path} $P$ is a path in $G$ such that $V(P) \cap S = \emptyset$ and $V(P) \cap N(S) = \{u, v\}$, where $u$ and $v$ are the endpoints of $P$.
We say that $S$ is a \emph{fen-core} if each cycle in $G$ has at least one edge in $G[S]$, there are at most two $(u, S)$-paths for each $u \in \overline{S}$, each $S$-outer path has length at least 7, and $|S| \le 32p_G$.
\end{definition}

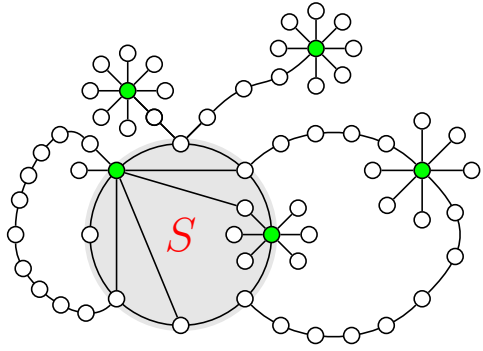
\begin{figure}[h]
\scalebox{1.5}{
\begin{tikzpicture}
\tikzset{
ff/.style = {draw,black, circle, fill=white, minimum width=4pt, inner sep=0pt}}    
\tikzset{st4/.style={postaction={decorate,
decoration={markings,
mark= at position 1/5 with {\node [ff]{};},
mark= at position 2/5 with {\node [ff]{};},
mark= at position 3/5 with {\node [ff]{};}, 
mark= at position 4/5 with {\node [ff]{};}, 
}}}}
\tikzset{st2/.style={postaction={decorate,
decoration={markings,
mark= at position 1/3 with {\node [ff]{};},
mark= at position 2/3 with {\node [ff]{};},
}}}}
\definecolor{Hcolor}{RGB}{182,255,193}
\tikzmath{ \cir = 0.8;}
\fill[gray!20!white] (0,0) circle (0.85 cm);
\draw (0,0) circle (\cir cm);
\def\n{8}
\node[red] at (0,0) {\large $S$};
\begin{scope}[every node/.style={draw, circle, minimum width=4pt, inner sep=0pt}]
\node[fill=white] (a1) at (360/\n*1: \cir cm) {};
\node[fill=white] (a2) at (360/\n*2: \cir cm) {};
\node[fill=green] (a3) at (360/\n*3: \cir cm) {};
\node[fill=white] (a4) at (360/\n*4: \cir cm) {};
\node[fill=white] (a5) at (360/\n*5: \cir cm) {};
\node[fill=white] (a6) at (360/\n*6: \cir cm) {};
\node[fill=white] (a7) at (360/\n*7: \cir cm) {};
\node[fill=green] (a8) at (360/\n*8: \cir cm) {};
\node[above left = 5pt of a2] (a21) {};
\node[above right = 5pt of a2] (a22) {};
\node[fill=green, above left = 5pt of a21] (a211) {};
\node[above right= 5pt of a211] (e2) {};
\draw (a211)--(e2);
\node[below right= 5pt of a211] (e3) {};
\draw (a211)--(e3);
\node[fill=white, left= 5pt of a211] (e4) {};
\draw (a211)--(e4);
\node[fill=white, above left= 5pt of a211] (e5) {};
\draw (a211)--(e5);
\node[fill=white, below left= 5pt of a211] (e6) {};
\draw (a211)--(e6);
\node[above= 5pt of a211] (e7) {};
\draw (a211)--(e7);
\node[below= 5pt of a211] (e8) {};
\draw (a211)--(e8);
\node[right= 5pt of a211] (e9) {};
\draw (a211)--(e9);
\draw (a21)--(a2)--(a22);    
\draw (a211)--(a21);   
\node[fill=green, right = 40pt of a1] (b1) {};
\node[right = 40pt of a7] (c1) {};      
\draw[st4] (a1) to[out=45,in=135, distance=15pt] (b1);
\draw[st2] (b1) to[out=-45,in=45, distance=15pt] (c1);
\draw[st4] (a7) to[out=-45,in=-135, distance=15pt] (c1);
\node[right= 5pt of a8] (d1) {};
\draw (a8)--(d1);
\node[above right= 5pt of a8] (d2) {};
\draw (a8)--(d2);
\node[below right= 5pt of a8] (d3) {};
\draw (a8)--(d3);
\node[fill=white, left= 5pt of a8] (d4) {};
\draw (a8)--(d4);
\node[fill=white, above left= 5pt of a8] (d5) {};
\draw (a8)--(d5);
\node[fill=white, below left= 5pt of a8] (d6) {};
\draw (a8)--(d6);
\node[right = 8pt of b1] (b3) {};
\draw (b1)--(b3);    
\node[above right = 8pt of b1] (b4) {};
\draw (b1)--(b4);     
\node[below = 8pt of b1] (b51) {};
\draw (b1)--(b51);
\node[left = 8pt of b1] (b31) {};
\draw (b1)--(b31);    
\node[below left = 8pt of b1] (b41) {};
\draw (b1)--(b41);     
\node[above = 8pt of b1] (b5) {};
\draw (b1)--(b5);
\node[fill=green, above right=20pt of a22, xshift=10pt] (b2) {};
\draw[st2] (a22) to[out =45, in = -135, distance=15] (b2);
\node[above right= 5pt of b2] (f2) {};
\draw (b2)--(f2);
\node[below right= 5pt of b2] (f3) {};
\draw (b2)--(f3);
\node[fill=white, above left= 5pt of b2] (f5) {};
\draw (b2)--(f5);
\node[fill=white, left= 5pt of b2] (f6) {};
\draw (b2)--(f6);
\node[above= 5pt of b2] (f7) {};
\draw (b2)--(f7);
\node[below= 5pt of b2] (f8) {};
\draw (b2)--(f8);
\node[right= 5pt of b2] (f9) {};
\draw (b2)--(f9);
\draw (a21)--(a2)--(a22);
\draw (a5)--(a3)--(a1);
\draw (a6)--(a3)--(d5);
\node[above left = 5pt of a3] (g) {};
\node[left = 5pt of a3] (g2) {};
\draw (g)--(a3)--(g2);
\node[left = 0.5cm of a4] (h) {};
\draw[st4] (g) to[out=135, in=90, distance = 10pt] (h);
\draw[st4] (a5) to[out=-135, in=-90, distance = 10pt] (h);
\end{scope}
\end{tikzpicture}}
\centering
\caption{A depiction of the fen-core $S$, see Definition~\ref{def:fen-core} (replicated from~\cite{DBLP:conf/mfcs/Balaban26} with permission). If $k \le 9$, then the green vertices may be used as $b$-vertices. Note that in contrast to this figure, $G[S]$ may be disconnected.}
\label{fig:fen-core}
\end{figure}

The fen-core can be easily found.

\begin{lemma}[{\cite[Lemma 18]{DBLP:conf/mfcs/Balaban26}}]\label{lem:S-exists}
A fen-core exists and it can be found in polynomial time.
\end{lemma}

Let us fix a fen-core $S$ and let $p = |S|$. Informally, a $\chi$-candidate is a vertex that can become a $b$-vertex in some supercoloring of $\chi$.

\begin{definition}[{\cite[Definition 19]{DBLP:conf/mfcs/Balaban26}}]\label{def:chi-candidates}
Let $U \seq V(G)$ and let $\chi$ be a proper coloring of $G[U]$.
The \emph{$\chi$-redundancy} of a vertex $v \in V(G)$ is the integer $\red_\chi(v) = |N[v]\setminus U| + |\chi(N[v])| - k$.
If $\red_\chi(u) \ge 0$, we say that $u$ is a \emph{$\chi$-candidate.}
If $\red_\chi(u) = 0$, we say that $u$ is \emph{$\chi$-tight.}
\end{definition}

Now we define a \emph{valid $S$-profile}, which is an object whose existence implies the existence of a $k$-$b^*$-coloring, see Lemma~\ref{lem:fen-find-coloring}.

\begin{definition}\label{def:fen-profile}
If $\chi\colon S \rightarrow [p]$ is a coloring of $G[S]$, $u \in V(G)$, and $B \seq S \cap N[u]$, then $(\chi, u, B)$ is an \emph{$S$-profile}.
An $S$-profile $(\chi, u, B)$ is \emph{valid} if $\chi$ is proper, each vertex in $B$ is a $\chi$-candidate, $|B| = |\chi(B)|$, $u \in S$ implies $u \in B$, and there are $k-|B|$ $\chi$-candidates in $N[u] \setminus S$.
\end{definition}

\begin{lemma}\label{lem:fen-find-coloring}
Given a valid $S$-profile $\zeta = (\chi, u, B)$, a $k$-$b^*$-coloring of $G$ can be computed in polynomial time.
\end{lemma}
\begin{proof}
Let $(\chi, u, B)$ be a valid $S$-profile.
Let $b = |B|$; we may assume that $\chi(B) = [b]$ (we permute the colors in $\chi$ if necessary).
Observe that $b \le p$.
Let $U \seq N[u] \setminus S$ be a set such that $|U| = k-b$ and every vertex in $U$ is a $\chi$-candidate ($U$ exists by Definition~\ref{def:fen-profile}).
Our algorithm will construct a $k$-$b^*$-coloring $\psi$ of $G$ such that $\chi \seq \psi$, $u$ is a $\psi$-$b^*$-vertex, and each vertex $v \in B \cup U$ is a $\psi$-$b$-vertex.

We begin by defining a coloring $\chi_1 \supseteq \chi$ so that $\dom(\chi_1) = S \cup U$ and $\chi_1(N[u]) = [k]$.
Let $U' = \{v \in U \sep N(v) \cap S \nsubseteq \{u\}\}$.
Observe that $|U'| \le 1$ since otherwise there would be an $S$-outer path of length at most 2, which would contradict Definition~\ref{def:fen-core}.
For the same reason, if $U' \ne \emptyset$, then $B = \emptyset$ (here we use the fact that $u \in S$ implies $u \in B$).

First, suppose that $U' = \{v\}$, and let $c \in \chi(N(v))$; in fact, $\chi(N(v)) = \{c\}$ because $v$ cannot have more than one neighbor in $S$.
Let $c_1, c_2 \in [k] \setminus \{c\}$ be distinct colors (they exist since $k \ge 4$).
Now we define $\chi_1(v) = c_1$, $\chi_1(u) = c_2$, and for $w \in U \setminus \{u ,v\}$, we define $\chi_1(w)$ arbitrarily so that $\chi_1(U) = [b+1, k]$.
Clearly, $\chi_1$ is proper, and every $w \in U \setminus \{v\}$ is a $\chi_1$-candidate.
Moreover, $v$ is a $\chi_1$-candidate since $\chi_1(u) \notin \chi(N(v))$.
Hence, all vertices in $B \cup U$ are $\chi_1$-candidates (recall that $B = \emptyset$ in this case).

Second, suppose that $U' = \emptyset$.
If $B = \emptyset$ or $u \in B$, then we define $\chi_1$ arbitrarily so that $\chi_1(U) = [b+1, k]$.
Otherwise, we have $u \notin B$ and $|B| = 1$ since $u$ cannot have more than one neighbor in $B \seq S$.
Let $\{v\} = B$.
Now we define $\chi_1(u)$ so that, if possible, $\chi_1(u) \notin \chi(N[v])$, and for $w \in U \setminus \{u\}$, we define $\chi_1(w)$ arbitrarily so that $\chi_1(U) = [b+1, k]$.
Again, $\chi_1$ is clearly proper, and every $w \in (B \cup U) \setminus \{v\}$ is a $\chi_1$-candidate.
Suppose for contradiction that $v$ is not a $\chi_1$-candidate.
In particular, we have $\chi_1(N[v]) \ne [k]$; see Definition~\ref{def:chi-candidates}.
Let $c \in [k] \setminus \chi_1(N[v])$.
Since $v$ is a $\chi$-candidate, there is a color $c' \in \chi(N[v]) \cap \chi_1(N[v] \cap U)$.
By Definition~\ref{def:fen-core}, $N(v) \cap U \seq \{u\}$ (otherwise there would an $S$-outer path of length 1).
Hence, $c' = \chi_1(u)$, which contradicts the construction of $\chi_1$ because $\chi_1(u)$ could be defined to be $c$ instead of $c'$.
Hence, all vertices in $B \cup U$ are $\chi_1$-candidates.

\medskip
Second, we define a partial coloring $\chi_2 \supseteq \chi_1$ such that $u$ is a $\chi_2$-$b^*$-vertex.
Let $v \in B \cup U$, $C_v = [k] \setminus \chi_1(N[v])$, and let $U_v \seq N(v)$ be such that $\chi_1(U_v) = \emptyset$ and $|U_v| = |C_v|$; informally, $v$ can be made into a $b$-vertex by assigning colors from $C_v$ to $U_v$.
We remark that $U_v$ exists because $v$ is a $\chi_1$-candidate.
Moreover, $U_v \cap U_w = \emptyset$ for all $w \in (B \cup U) \setminus \{v\}$ by Definition~\ref{def:fen-core}.
Let $U_v' = U_v \cap N(S \setminus \{v\})$ and observe that $|U_v'| \le 1$ by Definition~\ref{def:fen-core}.
If $U_v' = \emptyset$, then we assign colors in $C_v$ to vertices in $U_v$ arbitrarily.
Suppose that $w \in U_v'$ and observe that $v \notin N[S]$, $|N(w) \cap S| = 1$, and $N(v) \cap U \seq \{u\}$.
Let $x \in N(w) \cap S$ and let $c \in [k] \setminus \chi_1(\{u, v, x\})$ be a color; $c$ exists because $k \ge 4$.
Observe that $C_v = [k] \setminus \chi_1(\{u, v\})$; in particular, $c \in C_v$.
Now we define $\chi_2(w) = c$, and we assign the colors in $C_v \setminus \{c\}$ to vertices in $U_v \setminus \{w\}$ arbitrarily.
Hence, $\chi_2$ can be defined, and it is easy to see that every $v \in B \cup U$ is a $\chi_2$-$b$-vertex and $u$ is a $\chi_2$-$b^*$-vertex.

Suppose for contradiction that $wx \in E(G)$ and $c := \chi_2(w) = \chi_2(x)$.
Since $\chi_1$ is proper, we may assume that $w \in U_v$ for some $v \in B\cup U$ and $x \in N[B \cup U] \cup S$.
If $x \in N[B \cup U]$, then we obtain a contradiction with Definition~\ref{def:fen-core} (there would have to be a cycle of length at most 5 not contained in $G[S]$).
If $x \in S$, then $x \ne v$ (because $\chi_2(w) \in C_v$ by construction) and $w \in U_v'$, which means that we have a contradiction with the construction of $\chi_2$.
Hence, $\chi_2$ is proper.

\medskip
Finally, we define a total $k$-$b^*$-coloring $\psi \supseteq \chi_2$.
If $T$ is a connected component of $G$ such that $\chi_2(V(T)) = \emptyset$, then $T$ is a tree, and it is trivial to properly color $T$.
From now on, assume that there is no such component $T$.
We color the remaining vertices in a BFS order $\prec$ from $S_2 := \dom(\chi_2)$, i.e., if $v, w \in \overline{S_2}$ and $\dist(v, S_2) < \dist(w, S_2)$, then $v \prec w$.

Let $v \in \overline{S_2}$ and suppose that $\psi(w)$ is defined for all $w \prec v$.
Let $W = \{w \in N(v) \sep w \prec v$ or $w \in S_2\}$.
Suppose for contradiction that $w_1, w_2, w_3, w_4 \in W$ are distinct vertices.
For $i \in [4]$, let $P_i$ be a shortest $w_i$-$S_2$ path, and observe that $v \notin V(P_i)$.
By Definition~\ref{def:fen-core}, every cycle in $G$ must have an edge in $G[S]$, which means that $V(P_i) \cap V(P_j) = \emptyset$ for distinct $i,j\in [4]$.
Since there are at most two $v$-$S$ paths (again by Definition~\ref{def:fen-core}), we may assume, without loss of generality, that $V(P_i) \cap S = \emptyset$ for both $i \in [2]$.
Let us fix $i \in [2]$.
If $P_i$ has length 0, then we define $x_i = w_i$, and otherwise we define $x_i$ be the endpoint of $P_i$ distinct from $w_i$. Observe that $x_i$ belongs to $N[U]$.
Let $P_i'$ be the $v$-$u$ path obtained from $P_i$ by adding at most three edges as follows: we add $vw_i$ and if $x_i \ne u$, then either we add $x_iu$ if $x_i \in U$ or we add the two edges connecting $x_i$ and $u$.
If $V(P_i') \cap S \ne \emptyset$, then there would be three $v$-$S$ paths, namely $P_3$, $P_4$, and the $v$-$S$ subpath of $P_i'$, which would again be a contradiction.
Hence, $V(P_i') \cap S = \emptyset$, and the cycle obtained by combining $P_1'$ and $P_2'$ has no edge in $G[S]$, which is again contradiction.

We have proven $|W| \le 3$, which implies $|\chi_2(N(v))| \le 3$.
Since $k \ge 4$, there is a color which can be assigned to $v$.
After all such vertices $v$ are processed, we obtain the desired  $k$-$b^*$-coloring $\psi$.
\qed\end{proof}

Finally, we are ready to describe the desired algorithm.

\begin{theorem}
The \bscoloring problem, with input $(G, k)$, is \FPT parameterized by the feedback edge number $p$ of $G$, and it can be solved in time $2^{\ca O(p \log p)} \cdot n^{\ca O(1)}$, where $n= |V(G)|$.
\end{theorem}
\begin{proof}
If $k \le 3$, then we use the tree-width-based algorithm from Proposition~\ref{prop:ktw}.
Since the tree-width of $G$ is at most $p+1$, the running time in this case is $2^{\ca O(p)} \cdot n^2$.
From now on, suppose that $k \ge 4$.

First, we compute a fen-core $S$ of $G$ using Lemma~\ref{lem:S-exists}.
Second, we iterate over all $S$-profiles $(\chi, u, B)$, and for each of them, we test whether it is valid.
If we find a valid $S$-profile, we compute a $k$-$b^*$-coloring of $G$ using Lemma~\ref{lem:fen-find-coloring}. Otherwise, we answer that no such coloring exists.

Let us argue for correctness of the algorithm.
If no $k$-$b^*$-coloring exists, then we obviously cannot find one.
Conversely, suppose that $\psi$ is a $k$-$b^*$-coloring of $G$.
Let $\chi = \psi \upharpoonright S$, let $u$ be a $\psi$-$b^*$-vertex, and let $B_0 = \{v \in S \cap N[u]\sep v$ is a $\psi$-$b$-vertex$\}$.
For every color $c \in \chi(B_0)$, let $v_c \in B_0 \cap \psi^{-1}(c)$, and let $B = \{v_c \sep c \in \chi(B_0)\}$.
Since $\psi$ is proper, $\chi$ is proper as well.
Observe that every $\psi$-$b$-vertex is a $\chi$-candidate and that $|B| = |\chi(B)|$.
Finally, $u$ is the unique $\psi$-$b$-vertex for color $\psi(u)$ in $N[u]$, which means that $u \in S$ implies $u \in B$.
Hence, $(\chi,u,B)$ is a valid $S$-profile, and a $\psi$-$b^*$-vertex is found by Lemma~\ref{lem:fen-find-coloring}.

By Lemmas~\ref{lem:S-exists} and~\ref{lem:fen-find-coloring}, the running time is \FPT if we show that the number of $S$-profiles $(\chi, u, B)$ is small.
By Definition~\ref{def:fen-core}, we know that $|S| \in \ca O(p)$.
Hence, there are $p^{\ca O(p)}$ choices for $\chi$, $n$ choices for $u$, and $2^{\ca O(p)}$ choices for $B$.
Hence, there are $2^{\ca O(p \log p)} \cdot n \cdot 2^{\ca O(p)}$ $S$-profiles, which concludes the proof.
\qed\end{proof}

\subsection{Neighborhood diversity}

We will need the following lemma. In a nutshell, the algorithm described in the lemma works by constructing an ILP instance based on $B$.

\begin{lemma}[{\cite[implicit in the proof of Theorem 16]{JaffkeLS23}}]\label{lem:nd}
Let $G$ be an $n$-vertex graph, let $\ca P = (P_1, \ldots, P_d)$ be an ND-partition of $G$, let $B \seq [d]$ be such that $P_i$ induces a clique in $G$ for every $i \in B$, and let $P_B = \bigcup_{i\in B} P_i$.
Given $G$, $\ca P$, $B$, and an integer $k$, it can be decided, in time $2^{\ca O(d \log d)} \log n$, whether there is a $k$-$b$-coloring $\chi$ of $G$ such that 
\begin{compactenum}
\item all vertices in $P_B$ are $\chi$-$b$-vertices, and
\item for every color $c \in [k]$, $P_B$ contains a $\chi$-$b$-vertex of color $c$.
\end{compactenum}
We say that $B$ is the \emph{$b$-certificate set} for $\chi$.
\end{lemma}

Using Lemma~\ref{lem:nd}, we are able to prove the following.

\begin{theorem}\label{thm:nd}
The \bscoloring problem, with input $(G, k)$, is \FPT parameterized by the neighborhood diversity $d$ of $G$. Given an optimal ND-partition of $G$, the algorithm runs in time $2^{\ca O(d \log d)} \log n + \ca O(n)$, where $n = |V(G)|$.
\end{theorem}

\begin{proof}
Let $\ca P = (P_1, \ldots, P_d)$ be an ND-partition of $G$ with the smallest number of parts.
As in~\cite[Theorem 16]{JaffkeLS23}, we define $\ca P' = (P_1', \ldots, P'_{d'})$ to be the partition obtained from $\ca P$ by splitting every part $P_i$ that induces an independent set of size $s \ge 2$ into two parts: one part of size $s-1$ and one of size $1$.
Observe that $\ca P'$ is an ND-partition of $G$ and that $d' \le 2d$.

Let $H$ be the quotient graph $G/\ca P'$, i.e., $V(H) = [d']$ and for distinct $i,j \in [d']$, we have $ij \in E(H)$ if and only if there is an edge (or, equivalently, all possible edges) between $P_i'$ and $P_j'$ in $G$.
We say that a set $B \seq [d']$ is \emph{valid} if $P_i'$ induces a clique in $G$ for every $i \in B$ and there is an index $i \in B$ such that $ij \in E(H)$ for all $j \in B \setminus \{i\}$.
The algorithm iterates over all valid sets $B$ and for each of them, it decides using Lemma~\ref{lem:nd} whether there is a $k$-$b$-coloring $\chi$ of $G$ with a $b$-certificate set $B$.
If such a coloring exists for some valid set $B$, we output YES, and otherwise we output NO.

For a set $B \seq [d']$, we will use a shorthand $P_B = \bigcup_{i\in B} P_i'$.
First, suppose that for some valid set $B$, we find a $k$-$b$-coloring $\chi$ of $G$ with $b$-certificate set $B$.
Let $i \in B$ be such that $ij \in E(H)$ for all $j \in B \setminus \{i\}$, and let $u \in P_i'$.
Since $P_i'$ induces a clique in $G$, we have $P_B \seq N_G[u]$.
Since every color $c \in [k]$ has a $\chi$-$b$-vertex in $P_B$, we obtain that $u$ is a $\chi$-$b^*$-vertex, as desired.

Conversely, let $\chi$ be a $k$-$b^*$-coloring of $G$.
Let $u_k$ be a $\chi$-$b^*$-vertex; we may assume that $\chi(u_k) = k$.
Let $u_1, \ldots, u_{k-1} \in N(u_k)$ be such that $\chi(u_i) = i$ and $u_i$ is a $\chi$-$b$-vertex for every $i \in [k-1]$.
For every $j \in [k]$, let $i_j \in [d]$ be such that $u_j \in P_{i_j}$.
Suppose for contradiction that $i_j = i_\ell$ for some $1 \le j < \ell \le k$, and $P_{i_j}$ induces an independent set in $G$.
Since $u_j$ is a $\chi$-$b$-vertex, there is a vertex $v \in N(u_j) \cap \chi^{-1}(\ell)$.
By definition of an ND-partition, $u_{\ell}v \in E(G)$, which is a contradiction since $\chi$ is proper.
Hence, no such indices $j, \ell \in [k]$ exist.

For every $j \in [k]$, let $i_j' \in [d']$ be such that $u_j \in P_{i_j'}'$.
If $P_{i_j}$ induces an independent set in $G$ for some $j \in [k]$, then by definition of $\ca P'$, we may assume that $P_{i_j'}' = \{u_j\}$ (otherwise we may simply relabel vertices in $P_{i_j}$).
Hence, $P_{i_j'}'$ induces a clique in $G$ for every $j \in [k]$.
Let $B = \{i_j' \sep j \in [k]\}$.
Now it suffices to show that $B$ is a $b$-certificate set for $\chi$.
Suppose for contradiction that a vertex $v \in P_B$ is not a $\chi$-$b$-vertex, i.e., there is a color $c \in [k] \setminus \chi(N(v))$.
By definition of $B$, there is an index $j \in [k]$ such that $v$ and $u_j$ are twins and $vu_j \in E(G)$.
Hence, $c \ne j$ and $c \notin \chi(N(u_j))$, which is a contradiction since $u_j$ \emph{is} a $\chi$-$b$-vertex.
Hence, condition~1 stated in Lemma~\ref{lem:nd} is satisfied.
Moreover, condition~2 holds since $P_B$ contains a $\chi$-$b$-vertex for every color $c \in [k]$, namely $u_c$.
Hence, $B$ is indeed a $b$-certificate set for $\chi$.

We have shown the correctness of the algorithm. Finally, let us prove the bound on the running time. Observe that computing $\ca P'$ takes time $\ca O(n)$. Moreover, there are at most $2^{d'} \in 2^{\ca O(d)}$ valid sets $B$.
For each such set, Lemma~\ref{lem:nd} is used once, which concludes the proof.
\qed\end{proof}

\subsection{Twin cover number}

Let us fix a graph $G$ and an integer $k$. Let $S$ be a twin cover of $G$ and let $\chi$ be a coloring of $G[S]$.
For $A \seq S$, let $\ca K_A$ be the set of maximal cliques in $G-S$ whose neighborhood in $S$ is $A$, and let $c_A^\chi$ be the number of colors used on $A$ by $\chi$.
We say that $\chi$ is \emph{extendable} if it is proper and $|K| \le k - c_A^\chi$ for every clique $K \in \ca K_A$.
Assume that $\chi$ is extendable since otherwise there is no proper coloring of $G$ extending $\chi$.
The following statement is easy to observe.

\begin{observation}[{\cite[Observation 18, typo fixed]{JaffkeLS23}}]\label{obs:tc-b}
Let $K \in \ca K_A$. A vertex $v \in K$ is a $b$-vertex of some $k$-$b$-coloring of $G$ that extends $\chi$ if and only if $|K| = k - c^\chi_A$.
\end{observation}

Let $K_A^{max} \in \ca K_A$ be a clique of maximum cardinality.
The first reduction rule we need is the following.

\begin{rr}[{\cite[Reduction Rule 19]{JaffkeLS23}}]\label{rr:1}
If there exists $A \seq S$ such that $|\bigcup \ca K_A| \ge k - c_A^\chi + 1$, then delete a vertex $v \in \bigcup \ca K_A \setminus K_A^{max}$ from the graph.
\end{rr}

The following lemma shows the correctness of Reduction rule~\ref{rr:1}.

\begin{lemma}\label{lem:rr1}
If $G'$ is the graph obtained from $G$ by applying Reduction rule~\ref{rr:1}, then $G$ admits a $k$-$b^*$-coloring extending $\chi$ if and only if $G'$ does.
\end{lemma}
\begin{proof}
Let $v \in V(G)$ be the vertex deleted by Reduction rule~\ref{rr:1}, and let $A \seq S$ and $K \in \ca K_A$ be such that $v \in K$.
First, suppose that there is a $k$-$b^*$-coloring $\psi'$ of $G'$ that extends $\chi$.
Since $\chi$ is extendable, there are $|K \setminus \{v\}| + c_A^\chi \le k-1$ colors used on $N(v)$, which means that there is a color $c \in [k] \setminus \psi'(N(v))$.
Define $\psi = \psi' \cup \{(v, c)\}$.
Since every $\psi'$-$b^*$-vertex is also a $\psi$-$b^*$-vertex, $\psi$ is a $k$-$b^*$-coloring of $G$, as desired.

Second, suppose that there is a $k$-$b^*$-coloring $\psi$ of $G$ that extends $\chi$, and let $u \in V(G)$ be a $\psi$-$b^*$-vertex: if possible, choose $u$ so that for every color $c \in [k]$, there is a $\psi$-$b$-vertex $u_c \in N[u]$ such that $\psi(u_c) = c$ and $u_c \notin \bigcup \ca K_A \setminus K_A^{max}$.
Suppose for contradiction that \[C := \{c \in [k] \sep u_c \in K_c \text{ for some }K_c \in \ca K_A \setminus \{K_A^{max}\}\} \ne \emptyset.\]
Let us fix $c \in C$.
Since $u_c$ is a $\psi$-$b$-vertex, we have $|K_c| = k - c^\chi_A$ by Observation~\ref{obs:tc-b}.
Hence, $|K_c| = |K_A^{max}|$, so we have $\psi(K_A^{max}) = \psi(K_c) = [k] \setminus \psi(A)$.
Hence, every color $d \in \psi(K_c)$ has a $b$-vertex also in $K_A^{max}$.
Now either $u \in A$ or $u \in K_c$.
In the former case,  $u_c$ could have been chosen in $K_A^{max}$,
and in the latter case, $u$ and $u_d$ for every $d \in C$ could have been chosen in $K_A^{max}$.
In both cases, we obtain a contradiction.
Hence, for all $c \in [k]$, $u_c \notin \bigcup \ca K_A \setminus K_A^{max}$ and, in particular, $u_c \notin K$.
Let $U = \{u_c \sep c \in [k]\}$ and observe that $v \notin U$ because $v \in K$.

Let $\psi' = \psi \upharpoonright V(G')$.
If $u$ is a $\psi'$-$b^*$-vertex, then we are done.
Hence, suppose that it is not, i.e., $C := \{c \in [k] \sep u_c$ is not a $\psi'$-$b$-vertex$\} \ne \emptyset$.
Let $d = \psi(v)$.
Since $v \notin U$, we know that for every $c \in C$, $v$ is the unique vertex colored with $d$ by $\psi$ in $N(u_c)$.
Let $c \in C$ be a color.
Since $u_c \notin K$, we have $u_c \in A$.
Since $|\bigcup \ca K_A| \ge k - c_A^\chi + 1$ (see Reduction rule~\ref{rr:1}), we may use the pigeonhole principle to deduce that there are two vertices $x, y \in \bigcup \ca K_A$ such that $\psi(x) = \psi(y)$.
Since $\psi$ is proper, we may assume that $x \notin K_A^{max}$.
Let $\psi'' = \psi'[x \mapsto d]$.
Clearly, $u_c$ is a $\psi''$-$b$-vertex for every $c \in C$.
Moreover, for $c \in [k] \setminus C$, we have $u_c \notin \bigcup \ca K_A \setminus K_A^{max}$, which means that $u_c$ is a $\psi''$-$b$-vertex as well.
Hence, $u$ is a $\psi''$-$b^*$-vertex, as desired.
Finally, $\psi''$ is proper because $\bigcup \ca K_A \seq N(u_c)$ for any $c \in C$.
\qed\end{proof}

The second reduction rule is the following.

\begin{rr}[{\cite[Reduction Rule 22]{JaffkeLS23}}]\label{rr:2}
For some $A \seq S$, delete the cliques in $\ca K_A \setminus \{K_A^{max}\}$ and add a new clique of size $|\bigcup \ca K_A \setminus K_A^{max}|$ whose neighborhood in $G$ is exactly $A$.
\end{rr}

The following lemma shows the correctness of Reduction rule~\ref{rr:2}.

\begin{lemma}\label{lem:rr2}
If Reduction rule~\ref{rr:1} is no longer applicable to $G$ and $G'$ is the graph obtained from $G$ by applying Reduction rule~\ref{rr:2}, then $G$ admits a $k$-$b^*$-coloring extending $\chi$ if and only if $G'$ does.
\end{lemma}
\begin{proof}
We may assume that $V(G) = V(G')$ and $E(G) \seq E(G')$.
Let $A \seq S$ and let $W = \bigcup \ca K_A \setminus K_A^{max}$.
In addition, observe that by Observation~\ref{obs:tc-b}, no vertex in $W$ may be a $b$-vertex in any coloring extending $\chi$. Hence, if $\psi$ is a $k$-$b^*$-coloring of $G'$, then $\psi$ is also a $k$-$b^*$-coloring of $G$.

On the other hand, assume that $\psi$ is a $k$-$b^*$-coloring of $G$.
Since Reduction rule~\ref{rr:1} is no longer applicable to $G$, we know that $|\bigcup \ca K_A| \le k - c_A^\chi$.
Hence, we may recolor vertices in $W$ so that each of them has a different color not present in $\psi(A \cup K_A^{max})$.
The obtained coloring $\psi'$ is clearly a $k$-$b^*$-coloring of $G'$.
\qed\end{proof}

\begin{lemma}[{\cite[Lemma 23]{JaffkeLS23}}]\label{lem:tcn-nd}
When Reduction rules~\ref{rr:1} and~\ref{rr:2} are no longer applicable, the neighborhood diversity of $G$ is at most $2^{t+1} + t$.
\end{lemma}

\begin{theorem}
The \bscoloring problem, with input $(G, k)$, is \FPT parameterized by the twin cover number $t$ of $G$, and it can be solved in time $2^{2^{\ca O(t)}} n + \ca O(m)$, where $n = |V(G)|$ and $m = |E(G)|$.
\end{theorem}
\begin{proof}
First, we compute a minimum-size twin cover $S$ of $G$ in time $\ca O(1.2378^t + tn + m)$~\cite{ganian2011twin}.
Second, we guess an extendable $\min(k, t)$-coloring $\chi$ of $G[S]$.
Third, we exhaustively apply Reduction rule~\ref{rr:1} and then Reduction rule~\ref{rr:2}.
By Lemmas~\ref{lem:rr1} and~\ref{lem:rr2}, the obtained graph $G'$ admits a $k$-$b^*$-coloring extending $\chi$ if and only if $G$ does.
Finally, we apply Theorem~\ref{thm:nd} to $G'$: if the algorithm from Theorem~\ref{thm:nd} returns YES for some guess of $\chi$, then we return YES, and otherwise we return NO.

Observe that there are at most $t^t$ choices for $\chi$, both reduction rules can be applied efficiently, and $G'$ has neighborhood diversity at most $2^{t+1} + t$ by Lemma~\ref{lem:tcn-nd}.
Hence, by Theorem~\ref{thm:nd}, the total running time is as desired.
\qed\end{proof}

\subsection{Distance to co-cluster}

In this section, we will need some extra preliminaries.
The \emph{kernel} of a function $f\colon A \rightarrow B$, denoted $\ker(f)$, is the set $\{(a, a')\in A^2\sep f(a) = f(a')\}$.
The expression $\chi \circ h$ denotes the function $v \mapsto \chi(h(v))$.
Let us fix a graph $G$.
Two vertices $u$ and $v$ of $G$ are \emph{false twins} if they are twins and $uv \notin E(G)$.
If $H$ is a graph, then a function $h\colon G\rightarrow H$ is a \emph{surjective two-way homomorphism} if it is surjective and for every $u, v \in V(G)$, we have $uv \in E(G)$ if and only if $h(u)h(v) \in E(H)$.

\begin{lemma}\label{lem:two-way-hom}
Let $G$ and $H$ be graphs and let $h\colon G\rightarrow H$ be a surjective two-way homomorphism.
If $\chi$ is a $k$-$b^*$-coloring of $H$, then $\psi = \chi \circ h$ is a $k$-$b^*$-coloring of $G$.
Conversely, if $\psi$ is a $k$-$b^*$-coloring of $G$ such that $\ker(h) \seq \ker(\psi)$, then there is a $k$-$b^*$-coloring $\chi$ of $H$ such that $\psi = \chi \circ h$.
\end{lemma}
\begin{proof}
Let $G$, $H$, and $h$ be as in the statement.
Observe that $G$ is isomorphic to a graph that can be obtained from $H$ by adding one or more false twins to some vertices of $H$; hence, we may assume that $H$ is an induced subgraph of $G$. Indeed, for $u \in V(H)$, the set $h^{-1}(u) \seq V(G)$ consists of $u$ and all false twins added to $u$ (in particular, $h(u) = u$).
If $\chi$ is a $k$-$b^*$-coloring of $H$, then $\psi = \chi \circ h$ is clearly a proper coloring of $G$, and each $b$-vertex in $\chi$ is also a $b$-vertex in $\psi$ (for the same color).
Hence, a $\chi$-$b^*$-vertex is also a $\psi$-$b^*$-vertex, and $\psi$ is a $k$-$b^*$-coloring of $G$.

Conversely, assume that $\psi$ is a $k$-$b^*$-coloring of $G$ such that $\ker(h) \seq \ker(\psi)$.
Let us define $\chi = \psi \upharpoonright V(H)$; clearly, $\chi$ is a proper coloring of $H$ and $\psi = \chi \circ h$.
Let $u$ be a $\psi$-$b^*$-vertex and let $\{u_1, \ldots, u_{k-1}\} \seq N(u)$ be such that for every $i \in [k-1]$, $u_i$ is a $b$-vertex for color $i$.
Observe that for every $i \in [k-1]$, $h(u_i)\in V(H)\seq V(G)$ is also a $b$-vertex in $\psi$ because $u_i$ and $h(u_i)$ are twins in $G$ that are colored with the same color by $\psi$.
Clearly, $h(u_i)$ is a $b$-vertex also in $\chi$, and so $u$ is a $\chi$-$b^*$-vertex, as desired.
\qed\end{proof}

Let us fix an instance $(G, k)$ of \bscoloring.
Let $S \seq V(G)$ be a minimum co-cluster-modulator of $G$, let $p = \min(|S|, k)$, and let $\ca U$ be the set containing all maximal independent sets of $G - S$.
We will use $N^S(u)$ as a shorthand for $N(u) \cap S$.

Let us define ``types'' of sets in $\ca U$.

\begin{definition}[{\cite[Definition 8]{DBLP:conf/mfcs/Balaban26}}]\label{def:type}
A set $U \in \ca U$ has \emph{type} $t\colon 2^S \rightarrow [0, p+1]$ if for each $A \seq S$, we have $t(A) = \min(p+1, |\{u \in U \sep N^S(u) = A\}|)$. Let $T$ be the set of all types.
\end{definition}

Now we define a \emph{signature}, and then we define when a coloring has signature~$\sigma$.

\begin{definition}[{\cite[Definition 9]{DBLP:conf/mfcs/Balaban26}}]\label{def:signature}
A \emph{signature} is a tuple $\sigma = (\chi, q, \tau, \lambda, \xi)$, where:
\begin{compactitem}
\item $\chi\colon S \rightarrow [p]$ is a proper coloring of $G[S]$;
\item $q \in [0, p]$ is an integer;
\item $\tau \colon [q] \rightarrow T$, $\lambda\colon [p] \rightarrow [0, q]$, and $\xi\colon [p] \rightarrow 2^{2^S}$ are functions.
\end{compactitem}
A signature is required to have the following properties.
\begin{compactenum}
\item It holds that $[q]\seq \range(\lambda)$.\label{sign:surjective}
\item For each $t \in T$, $|\tau^{-1}(t)|$ is at most the number of sets of type $t$ in $\ca U$.\label{sign:set-types}
\item For each $i \in [q]$ and $A \seq S$, we have $|\{c \in [p]\colon \lambda(c) = i \land A \in \xi(c)\}| \le \tau(i)(A)$.\label{sign:vertex-types}
\end{compactenum}
\end{definition}

\begin{definition}[{\cite[Definition 10]{DBLP:conf/mfcs/Balaban26}}]\label{def:represent}
Let $\psi\colon V(G) \rightarrow [k]$ be a partial coloring and $\sigma = (\chi, q, \tau, \lambda, \xi)$ be a signature.
We say that $\psi$ \emph{has signature $\sigma$} (or that it is a \emph{$\sigma$-coloring})
if $\chi = \psi \upharpoonright S$ and:
\begin{compactenum}
\item there is a set $\ca U_\psi = \{U_1, \ldots, U_q\} \seq \ca U$ such that $\psi^{-1}([p]) \seq S\cup\bigcup \ca U_\psi$;\label{repr:blocks}
\item for each $i \in [q]$, $\tau(i)$ is the set-type of $U_i$;\label{repr:types}
\item for each $c \in [p]$, if $\lambda(c) = 0$, then $\psi^{-1}(c) \seq S$, and if $\lambda(c) = i \ne 0$, then $c \in \psi(U_i)$;\label{repr:f0}
\item for each $c \in [p]$, if $\lambda(c) = i \ne 0$, then $\xi(c) = \{A\seq S\sep\exists u\in U_i\colon  \psi(u) = c\land N^S(u) = A\}$;\label{repr:f1}
\end{compactenum}
\end{definition}

We say that $\psi$ is a \emph{minimal $\sigma$-coloring} if it is a $\sigma$-coloring, $\psi(V(G)) \seq [p]$, and for each $c \in [p]$ and $A \in \xi(c)$, there is a \emph{unique} vertex $u \in U_{\lambda(c)}$ such that $\psi(u) = c$ and $N^S(u) = A$.
It can be easily observed that minimal $\sigma$-colorings are exactly inclusion-wise minimal colorings that have signature $\sigma$.

Now we show that a minimal $\sigma$-coloring exists.

\begin{lemma}[{\cite[Lemma 11]{DBLP:conf/mfcs/Balaban26}}]\label{lem:minimal-sigma-exists}
For a signature $\sigma$, a minimal $\sigma$-coloring $\psi$ exists and can be computed in polynomial time. Moreover, the number of vertices colored by $\psi$ is $2^{\ca O(p)}$.
\end{lemma}

In the rest of the proof, we will need to refer to~\cite{balaban2025finding}, which is the full version of~\cite{DBLP:conf/mfcs/Balaban26}.
Informally, the following lemma says that two minimal $\sigma$-colorings are interchangeable, i.e., if one such coloring can be extended into a $k$-$b^*$-coloring of $G$, then any other such coloring can be extended as well.

\begin{lemma}\label{lem:any-minimal-coloring-can-be-used}
If $\sigma = (\chi_S, q, \tau, \lambda, \xi)$ is a signature, $\psi\colon V(G) \rightarrow [k]$ is a $b^*$-coloring of $G$ with signature $\sigma$, and $\chi \colon V(G) \rightarrow [p]$ is a minimal $\sigma$-coloring, then there is a $b^*$-coloring $\chi' \colon V(G) \rightarrow [k]$ of $G$ such that $\chi \seq \chi'$.
\end{lemma}
\begin{proof}
The proof is analogous to the proof of~\cite[Lemma 16]{balaban2025finding}.
The only difference is that instead of~\cite[Lemma 3]{balaban2025finding}, we need to use Lemma~\ref{lem:two-way-hom}.
\qed\end{proof}

Let us now observe that the behavior of colors \emph{not} in $[p]$ is significantly restricted in any $b$-coloring of $G$.

\begin{observation}[{\cite[Observation 15]{balaban2025finding}}]\label{obs:anonym}
If $\psi$ is a $k$-$b$-coloring of $G$ with signature $\sigma$, then for distinct colors $c, d \in [p+1, k]$, there are sets $U_c, U_d \in \ca U$ such that $\psi^{-1}(c) \seq U_c$, $\psi^{-1}(d) \seq U_d$, and $U_c \ne U_d$.
\end{observation}

The following simple lemma concerns two $\sigma$-colorings that are comparable by inclusion.

\begin{lemma}[{\cite[Lemma 17]{balaban2025finding}}]\label{lem:twins}
If $\sigma = (\chi_S, q, \tau, \lambda, \xi)$ is a signature and $\chi, \psi$ are two proper $\sigma$-colorings such that $\chi\seq\psi$, then $\ca U_\psi = \ca U_\chi$ (see Definition~\ref{def:represent}).
Moreover, each vertex $u\in V(G)$ such that $\psi(u) \in [p]$ has a twin $u'$ such that $\chi(u') = \psi(u)$.
\end{lemma}

In the following definition, a $\chi$-candidate pair $(u, B)$ is a pair such that $u$ should become a $b^*$-vertex and $B \seq N[u]$ should contain the $b$-vertices for colors in $[p]$.
In addition, $(\chi, u, B)$-candidates are those sets in $\ca U$ that may contain a $b$-vertex for some color in $[p+1, k]$, and $\chi$-flexible are those that do not have to contain such a $b$-vertex.

\begin{definition}\label{def:candidates-and-flexible}
Let $\sigma = (\chi_S, q, \tau, \lambda, \xi)$ be a signature and $\chi\colon V(G) \rightarrow [p]$ be a minimal $\sigma$-coloring.
We say that $(u, B)$ is a \emph{$\chi$-candidate pair} if $u \in V(G)$, $B \seq N[u]$, $|B| = p$, $\chi(B) = [p]$, and for each $v \in B$, it holds that $[p] \seq \chi(N[v])$.

Let us fix a $\chi$-candidate pair $(u, B)$.
We say that $U \in \ca U$ is a \emph{$(\chi, u, B)$-candidate} if each $v \in B$ has a neighbor $w \in U$ that is uncolored by $\chi$, and
there is a vertex $v\in U \cap N[u]$ that is uncolored by $\chi$ such that $[p] \seq \chi(N(v))$.

We say that $U$ is \emph{$\chi$-flexible} if each $v\in U$ uncolored by $\chi$ has a twin $w \in U$ that is colored by $\chi$.
\end{definition}

Now we characterize when a $b^*$-coloring exists, using the terminology introduced in Definition~\ref{def:candidates-and-flexible}.

\begin{lemma}\label{lem:find-coloring}
Let $\sigma = (\chi_S, q, \tau, \lambda, \xi)$ be a signature and $\chi\colon V(G) \rightarrow [p]$ be a minimal $\sigma$-coloring. There is a $b^*$-coloring $\psi\colon V(G) \rightarrow [k]$ of $G$ with signature $\sigma$ such that $\chi \seq \psi$ if and only if $\chi$ is proper and there is a $\chi$-candidate pair $(u, B)$ and a set $\ca C\seq\ca U$ such that $|\ca C| = k-p$, all sets in $\ca C$ are $(\chi, u, B)$-candidates, and all sets in $\ca U\setminus\ca C$ are $\chi$-flexible. Moreover, given $B$ and $\ca C$ satisfying these properties, the $b^*$-coloring $\psi$ can be computed in polynomial time.
\end{lemma}
\begin{proof}
First, we prove the left-to-right implication.
Let $\psi\colon V(G) \rightarrow [k]$ be a $b^*$-coloring of $G$ with signature $\sigma$ such that $\chi \seq \psi$. Since $\psi$ is proper, $\chi$ is proper as well.
Let $u \in V(G)$ be a $\psi$-$b^*$-vertex, and let $B := \{u_1, \ldots, u_p\} \seq N[u]$ be such that for every $i \in [p]$, $\psi(u_i) = i$ and $u_i$ is a $\psi$-$b$-vertex.
Let us fix $i \in [p]$.
By Lemma~\ref{lem:twins}, we may assume that $\chi(u_i) = i$.
Moreover, if $c \in [p] \setminus \{i\}$, then $c \in \psi(N(u_i))$ since $u_i$ is a $\psi$-$b$-vertex, and by Lemma~\ref{lem:twins}, we have $c \in \chi(N(u_i))$.
Hence, $[p] \seq \chi(N[u_i])$, which means that $(u, B)$ is a $\chi$-candidate pair.

Let $\ca C = \{U_c \sep c\in [p+1,k]\}$, where $U_c \in \ca U$ is the set such that $\psi^{-1}(c) \seq U_c$; it is unique by Observation~\ref{obs:anonym}.
Recall that the observation also says that $U_c\ne U_d$ for distinct colors $c,d\in[p+1,k]$, which easily implies that $|\ca C| = k-p$. Now let $U_c \in \ca C$. We need to prove that $U_c$ is a $(\chi, u, B)$-candidate.
First, let $u_i \in B$.
Since $u_i$ is a $\psi$-$b$-vertex, we have $c \in \psi(N(u_i))$, and since $c \notin \range(\chi)$, $u_i$ has a neighbor in $U_c$ that is uncolored by $\chi$.
Second, let $v \in U_c \cap N[u]$ be a $b$-vertex for $c$ in $\psi$ (it exists because $u$ is a $\psi$-$b^*$-vertex).
Let $d \in [p]$, and let $w \in N(v)$ be a vertex such that $\psi(w) = d$. By Lemma~\ref{lem:twins}, $w$ has a twin $w'$ such that $\chi(w') = d$. Hence, $[p] \seq \chi(N(v))$, and $U_c$ is indeed a $(\chi, u, B)$-candidate.
Now let $U \in \ca U\setminus\ca C$.
By definition of $\ca C$, we know that $\psi(U)\seq[p]$.
If $v \in U$ is a vertex uncolored by $\chi$, then by Lemma~\ref{lem:twins}, $v$ has a twin $v'$ that is colored by $\chi$, which shows that $U$ is $\chi$-flexible.

Now we prove the right-to-left implication. Let $u \in V(G)$, $B\seq N[u]$, and $\ca C\seq\ca U$ be as in the statement. We may assume that $\ca C = \{U_c \sep c\in [p+1,k]\}$, i.e., we label the sets in $\ca C$ arbitrarily.
Now we construct $\psi \supseteq \chi$ as follows.
Let $v \in V(G)$ be a vertex uncolored by $\chi$.
By Definition~\ref{def:represent}, $v\notin S$. If $v\in U_c \in\ca C$, we let $\psi(v) = c$, and if $v\in U \in\ca U \setminus\ca C$, then by Definition~\ref{def:candidates-and-flexible}, $v$ has a twin $v'\in U$ that is colored by $\chi$, and we let $\psi(v) = \chi(v')$.
Observe that this coloring $\psi$ can be computed in polynomial time as required.
Suppose there is an edge $vw \in E(G)$ such that $\psi(v) = \psi(w) = c$.
If $c \in [p]$, then by Lemma~\ref{lem:twins}, there would be an edge $v'w' \in E(G)$ such that $\chi(v') = \chi(w') = c$, which would be a contradiction with $\chi$ being proper.
On the other hand, if $c\in[p+1, k]$, then $v,w \in U_{c}$ by construction of $\psi$, a contradiction with $vw \in E(G)$. Hence, $\psi$ is a proper coloring.

Now we show that $\psi$ is a $b^*$-coloring. First, let $c \in [p]$ and let $v \in B$ be the unique vertex in $B$ such that $\psi(v) = c$.
Let $d \in [k]$ be a color. If $d \in [p]$, then by Definition~\ref{def:candidates-and-flexible}, $d \in \chi(N[v]) \seq \psi(N[v])$.
On the other hand, if $d \in [p+1, k]$, then $U_d \in \ca C$ is a $(\chi, u, B)$-candidate, which means that there is $w \in U_d \cap N(v)$ that is uncolored by $\chi$. However, by construction of $\psi$, we know that $\psi(w) = d$. Hence, $v$ is a $b$-vertex for $c$ in $\psi$.
Second, let $c\in [p+1,k]$. Since $U_c$ is a $(\chi, u, B)$-candidate, there is a vertex $v \in U_c \cap N[u]$ such that $[p] \seq \chi(N(v)) \seq \psi(N(v))$. If $d \in [p+1,k]$ and $d \ne c$, then $U_d \ne U_c$, and so $v$ has a neighbor of color $d$ in $\psi$. Hence, $v$ is a $b$-vertex for $c$ in $\psi$.
We have shown that $N[u]$ contains a $\psi$-$b$-vertex of every color, so $u$ is a $\psi$-$b^*$-vertex, as desired.
\qed\end{proof}

Finally, we are ready to present the desired algorithm.

\begin{theorem}
The \bscoloring problem, with input $(G, k)$, is \FPT parameterized by the distance to co-cluster $d$ of $G$, and it can be solved in time $2^{2^{\ca O(p)}} \cdot n^{\ca O(1)}$.
\end{theorem}
\begin{proof}
First, we compute a minimum co-cluster modulator $S$ of $G$~\cite{boral2016fast}.
Second, we try all signatures $\sigma$, and we find a minimal $\sigma$-coloring $\chi$ of $G$ using Lemma~\ref{lem:minimal-sigma-exists}.
If $\chi$ is not proper, we discard it and continue with the next signature. Otherwise, we try all $\chi$-candidate pairs $(u, B)$, and we compute for each set $U\in \ca U$, whether it is $\chi$-flexible or a $(\chi, u, B)$-candidate.
If there is a set $U\in\ca U$ that is neither $\chi$-flexible nor a $(\chi, u, B)$-candidate, we reject $(u, B)$. Otherwise, we compute the sets $\ca C_0 = \{U\in\ca U\sep U$ is a $(\chi, u, B)$-candidate but not $\chi$-flexible$\}$ and $\ca C_1 = \{U\in\ca U\sep U$ is a $\chi$-flexible $(\chi, u, B)$-candidate$\}$. 
If $|\ca C_0| > k-p$ or $|\ca C_0| + |\ca C_1| < k-p$, we reject $(u, B)$. Otherwise, we can find a set $\ca C$ such that $|\ca C| = k-p$ and $\ca C_0 \seq\ca C\seq\ca C_0\cup\ca C_1$, which allows us to compute a $b^*$-coloring of $G$ by Lemma~\ref{lem:find-coloring}. If we do not find a $b^*$-coloring for any $\sigma$ and $(u, B)$, we reject $(G, k)$.

The correctness and the bound on the running time of the algorithm can be proven analogously to~\cite[Theorem 7]{DBLP:conf/mfcs/Balaban26}.
The only difference is that the running time is worse by a multiplicative factor of $|V(G)|$ because we guess not only $B$ but also $u \in V(G)$.
\qed\end{proof}

\bibliographystyle{splncs04}
\bibliography{bibliography}

@inproceedings{DBLP:conf/mfcs/Balaban26,
  author       = {Jakub Balab{\'{a}}n},
  editor       = {Michal Kouck{\'{y}} and
                  Daniela Petrisan},
  title        = {Finding b-Colorings Using Feedback Edges},
  booktitle    = {51st International Symposium on Mathematical Foundations of Computer
                  Science, {MFCS} 2026, Paris, France, August 24-28, 2026},
  series       = {LIPIcs},
  volume       = {386},
  pages        = {36:1--36:18},
  publisher    = {Schloss Dagstuhl - Leibniz-Zentrum f{\"{u}}r Informatik},
  year         = {2026},
  doi          = {10.4230/LIPICS.MFCS.2026.36},
  bibsource    = {dblp computer science bibliography, https://dblp.org}
}

@article{elberfeld2015space,
  title={On the space and circuit complexity of parameterized problems: Classes and completeness},
  author={Elberfeld, Michael and Stockhusen, Christoph and Tantau, Till},
  journal={Algorithmica},
  volume={71},
  number={3},
  pages={661--701},
  year={2015},
  publisher={Springer}
}

@article{bodlaender2024parameterized,
  title={Parameterized problems complete for nondeterministic FPT time and logarithmic space},
  author={Bodlaender, Hans L and Groenland, Carla and Nederlof, Jesper and Swennenhuis, C{\'e}line},
  journal={Information and Computation},
  volume={300},
  pages={105195},
  year={2024},
  publisher={Elsevier}
}

@book{flum2006parameterized,
  title={Parameterized complexity theory},
  author={Flum, J{\"o}rg and Grohe, Martin},
  year={2006},
  publisher={Springer}
}

@article{rao2008clique,
  title={Clique-width of graphs defined by one-vertex extensions},
  author={Rao, Micha{\"e}l},
  journal={Discrete Mathematics},
  volume={308},
  number={24},
  pages={6157--6165},
  year={2008},
  publisher={Elsevier}
}

@book{CyganFKLMPPS15,
  author       = {Marek Cygan and
                  Fedor V. Fomin and
                  Lukasz Kowalik and
                  Daniel Lokshtanov and
                  D{\'{a}}niel Marx and
                  Marcin Pilipczuk and
                  Michal Pilipczuk and
                  Saket Saurabh},
  title        = {Parameterized Algorithms},
  publisher    = {Springer},
  year         = {2015},
  doi          = {10.1007/978-3-319-21275-3},
  isbn         = {978-3-319-21274-6},
  bibsource    = {dblp computer science bibliography, https://dblp.org}
}

@incollection{hall1987representatives,
  title={On representatives of subsets},
  author={Hall, Philip},
  booktitle={Classic papers in Combinatorics},
  pages={58--62},
  year={1987},
  publisher={Springer}
}

@article{zaker2026improved,
  title={Improved bounds on the b-chromatic number using the independence and chromatic numbers},
  author={Zaker, Manouchehr},
  journal={arXiv preprint arXiv:2606.07461},
  year={2026}
}

@article{hoang2005b,
  title={On the b-dominating coloring of graphs},
  author={Ho{\`a}ng, Ch{\'\i}nh T and Kouider, Mekkia},
  journal={Discrete Applied Mathematics},
  volume={152},
  number={1-3},
  pages={176--186},
  year={2005},
  publisher={Elsevier}
}

@article{jakovac2018b,
  title={The b-chromatic number and related topics—a survey},
  author={Jakovac, Marko and Peterin, Iztok},
  journal={Discrete Applied Mathematics},
  volume={235},
  pages={184--201},
  year={2018},
  publisher={Elsevier}
}

@article{barth2007b,
  title={On the b-continuity property of graphs},
  author={Barth, Dominique and Cohen, Johanne and Faik, Taoufik},
  journal={Discrete Applied Mathematics},
  volume={155},
  number={13},
  pages={1761--1768},
  year={2007},
  publisher={Elsevier}
}

@article{bonomo2009b,
  title={On the b-coloring of cographs and P 4-sparse graphs},
  author={Bonomo, Flavia and Dur{\'a}n, Guillermo and Maffray, Frederic and Marenco, Javier and Valencia-Pabon, Mario},
  journal={Graphs and Combinatorics},
  volume={25},
  number={2},
  pages={153--167},
  year={2009},
  publisher={Springer}
}

@article{ikhlef2010characterization,
  title={Characterization of some b-chromatic edge critical graphs},
  author={Ikhlef-Eschouf, Noureddine},
  journal={Australasian J. Combin},
  volume={47},
  pages={21--35},
  year={2010}
}

@article{jakovac2012b,
  title={On the b-chromatic number of some graph products},
  author={Jakovac, Marko and Peterin, Iztok},
  journal={Studia scientiarum mathematicarum Hungarica},
  volume={49},
  number={2},
  pages={156--169},
  year={2012},
  publisher={Akad{\'e}miai Kiad{\'o}}
}

@article{maffray2013b,
  title={b-colouring the Cartesian product of trees and some other graphs},
  author={Maffray, Fr{\'e}d{\'e}ric and Silva, Ana},
  journal={Discrete Applied Mathematics},
  volume={161},
  number={4-5},
  pages={650--669},
  year={2013},
  publisher={Elsevier}
}

@article{balakrishnan2014b,
  title={b-chromatic number of Cartesian product of some families of graphs},
  author={Balakrishnan, Rangaswami and Raj, S Francis and Kavaskar, T},
  journal={Graphs and Combinatorics},
  volume={30},
  number={3},
  pages={511--520},
  year={2014},
  publisher={Springer}
}

@article{kouider2007b,
  title={The b-chromatic number of the Cartesian product of two graphs},
  author={Kouider, Mekkia and Mah{\'e}o, Maryvonne},
  journal={Studia Scientiarum Mathematicarum Hungarica},
  volume={44},
  number={1},
  pages={49--55},
  year={2007},
  publisher={Akad{\'e}miai Kiad{\'o}}
}

@article{shaebani2019note,
  title={A note on b-coloring of Kneser graphs},
  author={Shaebani, Saeed},
  journal={Discrete Applied Mathematics},
  volume={257},
  pages={368--369},
  year={2019},
  publisher={Elsevier}
}

@article{balakrishnan2012b,
  title={b-coloring of Kneser graphs},
  author={Balakrishnan, Rangaswami and Kavaskar, T},
  journal={Discrete applied mathematics},
  volume={160},
  number={1-2},
  pages={9--14},
  year={2012},
  publisher={Elsevier}
}

@article{hajiabolhassan2010b,
  title={On the b-chromatic number of Kneser graphs},
  author={Hajiabolhassan, Hossein},
  journal={Discrete Applied Mathematics},
  volume={158},
  number={3},
  pages={232--234},
  year={2010},
  publisher={Elsevier}
}

@article{javadi2009b,
  title={On b-coloring of the Kneser graphs},
  author={Javadi, Ramin and Omoomi, Behnaz},
  journal={Discrete Mathematics},
  volume={309},
  number={13},
  pages={4399--4408},
  year={2009},
  publisher={Elsevier}
}

@article{el2014b,
  title={The b-chromatic number and f-chromatic vertex number of regular graphs},
  author={El Sahili, Amine and Kheddouci, Hamamache and Kouider, Mekkia and Mortada, Maidoun},
  journal={Discrete Applied Mathematics},
  volume={179},
  pages={79--85},
  year={2014},
  publisher={Elsevier}
}

@article{shaebani2012b,
  title={On the b-chromatic number of regular graphs without 4-cycle},
  author={Shaebani, Saeed},
  journal={Discrete Applied Mathematics},
  volume={160},
  number={10-11},
  pages={1610--1614},
  year={2012},
  publisher={Elsevier}
}

@article{jakovac2010b,
  title={The b-chromatic number of cubic graphs},
  author={Jakovac, Marko and Klav{\v{z}}ar, Sandi},
  journal={Graphs and Combinatorics},
  volume={26},
  number={1},
  pages={107--118},
  year={2010},
  publisher={Springer}
}

@article{SahiliKM15,
  author       = {Amine El Sahili and
                  Mekkia Kouider and
                  Maidoun Mortada},
  title        = {On the b-chromatic number of regular bounded graphs},
  journal      = {Discret. Appl. Math.},
  volume       = {193},
  pages        = {174--179},
  year         = {2015},
  url          = {https://doi.org/10.1016/j.dam.2015.04.017},
  doi          = {10.1016/J.DAM.2015.04.017},
  bibsource    = {dblp computer science bibliography, https://dblp.org}
}

@article{CabelloJ11,
  author       = {Sergio Cabello and
                  Marko Jakovac},
  title        = {On the b-chromatic number of regular graphs},
  journal      = {Discret. Appl. Math.},
  volume       = {159},
  number       = {13},
  pages        = {1303--1310},
  year         = {2011},
  url          = {https://doi.org/10.1016/j.dam.2011.04.028},
  doi          = {10.1016/J.DAM.2011.04.028},
  bibsource    = {dblp computer science bibliography, https://dblp.org}
}

@article{CamposLS15,
  author       = {Victor A. Campos and
                  Carlos Vin{\'{\i}}cius G. C. Lima and
                  Ana Silva},
  title        = {Graphs of girth at least 7 have high b-chromatic number},
  journal      = {Eur. J. Comb.},
  volume       = {48},
  pages        = {154--164},
  year         = {2015},
  doi          = {10.1016/J.EJC.2015.02.017},
  bibsource    = {dblp computer science bibliography, https://dblp.org}
}

@article{CamposSMS09,
  author       = {Victor A. Campos and
                  Cl{\'{a}}udia Linhares Sales and
                  Fr{\'{e}}d{\'{e}}ric Maffray and
                  Ana Silva},
  title        = {b-chromatic number of cacti},
  journal      = {Electron. Notes Discret. Math.},
  volume       = {35},
  pages        = {281--286},
  year         = {2009},
  doi          = {10.1016/J.ENDM.2009.11.046},
  bibsource    = {dblp computer science bibliography, https://dblp.org}
}

@article{CamposLMSSS15,
  author       = {Victor A. Campos and
                  Carlos Vin{\'{\i}}cius G. C. Lima and
                  Nicolas Almeida Martins and
                  Leonardo Sampaio Rocha and
                  Marcio Costa Santos and
                  Ana Silva},
  title        = {The b-chromatic index of graphs},
  journal      = {Discret. Math.},
  volume       = {338},
  number       = {11},
  pages        = {2072--2079},
  year         = {2015},
  doi          = {10.1016/J.DISC.2015.04.026},
  bibsource    = {dblp computer science bibliography, https://dblp.org}
}

@article{BonomoSSV15,
  author       = {Flavia Bonomo and
                  Oliver Schaudt and
                  Maya Stein and
                  Mario Valencia{-}Pabon},
  title        = {b-Coloring is NP-hard on Co-bipartite Graphs and Polytime Solvable
                  on Tree-Cographs},
  journal      = {Algorithmica},
  volume       = {73},
  number       = {2},
  pages        = {289--305},
  year         = {2015},
  doi          = {10.1007/S00453-014-9921-5},
  bibsource    = {dblp computer science bibliography, https://dblp.org}
}

@article{HavetSS12,
  author       = {Fr{\'{e}}d{\'{e}}ric Havet and
                  Cl{\'{a}}udia Linhares Sales and
                  Leonardo Sampaio Rocha},
  title        = {b-coloring of tight graphs},
  journal      = {Discret. Appl. Math.},
  volume       = {160},
  number       = {18},
  pages        = {2709--2715},
  year         = {2012},
  doi          = {10.1016/J.DAM.2011.10.017},
  bibsource    = {dblp computer science bibliography, https://dblp.org}
}

@inproceedings{KratochvilTV02,
  author       = {Jan Kratochv{\'{\i}}l and
                  Zsolt Tuza and
                  Margit Voigt},
  editor       = {Ludek Kucera},
  title        = {On the b-Chromatic Number of Graphs},
  booktitle    = {Graph-Theoretic Concepts in Computer Science, 28th International Workshop,
                  {WG} 2002, Cesky Krumlov, Czech Republic, June 13-15, 2002, Revised
                  Papers},
  series       = {Lecture Notes in Computer Science},
  volume       = {2573},
  pages        = {310--320},
  publisher    = {Springer},
  year         = {2002},
  doi          = {10.1007/3-540-36379-3\_27},
  bibsource    = {dblp computer science bibliography, https://dblp.org}
}

@book{Kloks94,
  author       = {Ton Kloks},
  title        = {Treewidth, Computations and Approximations},
  series       = {Lecture Notes in Computer Science},
  volume       = {842},
  publisher    = {Springer},
  year         = {1994},
  doi          = {10.1007/BFB0045375},
  isbn         = {3-540-58356-4},
  bibsource    = {dblp computer science bibliography, https://dblp.org}
}

@inproceedings{Korhonen21,
  author       = {Tuukka Korhonen},
  title        = {A Single-Exponential Time 2-Approximation Algorithm for Treewidth},
  booktitle    = {62nd {IEEE} Annual Symposium on Foundations of Computer Science, {FOCS}
                  2021, Denver, CO, USA, February 7-10, 2022},
  pages        = {184--192},
  publisher    = {{IEEE}},
  year         = {2021},
  doi          = {10.1109/FOCS52979.2021.00026},
  bibsource    = {dblp computer science bibliography, https://dblp.org}
}

@article{boral2016fast,
  title={A fast branching algorithm for cluster vertex deletion},
  author={Boral, Anudhyan and Cygan, Marek and Kociumaka, Tomasz and Pilipczuk, Marcin},
  journal={Theory of Computing Systems},
  volume={58},
  number={2},
  pages={357--376},
  year={2016},
  publisher={Springer}
}

@inproceedings{ganian2011twin,
  title={Twin-cover: Beyond vertex cover in parameterized algorithmics},
  author={Ganian, Robert},
  booktitle={International Symposium on Parameterized and Exact Computation},
  pages={259--271},
  year={2011},
  organization={Springer}
}

@article{JaffkeLL24,
  author       = {Lars Jaffke and
                  Paloma T. Lima and
                  Daniel Lokshtanov},
  title        = {b-Coloring Parameterized by Clique-Width},
  journal      = {Theory Comput. Syst.},
  volume       = {68},
  number       = {4},
  pages        = {1049--1081},
  year         = {2024},
  doi          = {10.1007/S00224-023-10132-0},
  bibsource    = {dblp computer science bibliography, https://dblp.org}
}

@article{balaban2025finding,
  title={Finding $ b $-colorings Using Feedback Edges},
  author={Balab{\'a}n, Jakub},
  journal={arXiv preprint arXiv:2512.14390},
  year={2025}
}

@article{DBLP:journals/jcss/PanolanPS17,
  author       = {Fahad Panolan and
                  Geevarghese Philip and
                  Saket Saurabh},
  title        = {On the parameterized complexity of b-chromatic number},
  journal      = {J. Comput. Syst. Sci.},
  volume       = {84},
  pages        = {120--131},
  year         = {2017},
  doi          = {10.1016/J.JCSS.2016.09.012},
  bibsource    = {dblp computer science bibliography, https://dblp.org}
}

@inproceedings{JaffkeLS23,
  author       = {Lars Jaffke and
                  Paloma T. Lima and
                  Roohani Sharma},
  editor       = {Satoru Iwata and
                  Naonori Kakimura},
  title        = {Structural Parameterizations of b-Coloring},
  booktitle    = {34th International Symposium on Algorithms and Computation, {ISAAC}
                  2023, December 3-6, 2023, Kyoto, Japan},
  series       = {LIPIcs},
  volume       = {283},
  pages        = {40:1--40:14},
  publisher    = {Schloss Dagstuhl - Leibniz-Zentrum f{\"{u}}r Informatik},
  year         = {2023},
  doi          = {10.4230/LIPICS.ISAAC.2023.40},
  bibsource    = {dblp computer science bibliography, https://dblp.org}
}

@book{Diestel,
  author       = {Reinhard Diestel},
  title        = {Graph Theory, 4th Edition},
  series       = {Graduate texts in mathematics},
  volume       = {173},
  publisher    = {Springer},
  year         = {2012},
  isbn         = {978-3-642-14278-9},
  bibsource    = {dblp computer science bibliography, https://dblp.org}
}

@article{BLIDIA20091787,
title = {On b-colorings in regular graphs},
journal = {Discrete Applied Mathematics},
volume = {157},
number = {8},
pages = {1787-1793},
year = {2009},
issn = {0166-218X},
doi = {10.1016/j.dam.2009.01.007},
author = {Mostafa Blidia and Frédéric Maffray and Zoham Zemir}
}

@article{kouider2004b,
  title={b-Chromatic number of a graph, subgraphs and degrees},
  author={Kouider, M},
  journal={Rapport interne LRI},
  volume={1392},
  year={2004}
}

@article{elsahili2009about,
  title={About b-colouring of regular graphs},
  author={El Sahili, Amine and Kouider, Mekkia},
  journal={Utilitas Mathematica},
  volume={80},
  pages={211--215},
  year={2009},
  publisher={Utilitas Mathematica Publishing}
}

@article{ZAKER2025370,
title = {On z-coloring and b*-coloring of graphs as improved variants of the b-coloring},
journal = {Discrete Applied Mathematics},
volume = {377},
pages = {370-379},
year = {2025},
issn = {0166-218X},
doi = {10.1016/j.dam.2025.07.036},
author = {Manouchehr Zaker}
}

@article{DETTLAFF2024128914,
title = {A new approach to b-coloring of regular graphs},
journal = {Applied Mathematics and Computation},
volume = {481},
pages = {128914},
year = {2024},
issn = {0096-3003},
doi = {10.1016/j.amc.2024.128914},
author = {Magda Dettlaff and Hanna Furmańczyk and Iztok Peterin and Adriana Roux and Radosław Ziemann}
}

@article{IRVING1999127,
title = {The b-chromatic number of a graph},
journal = {Discrete Applied Mathematics},
volume = {91},
number = {1},
pages = {127-141},
year = {1999},
issn = {0166-218X},
doi = {10.1016/S0166-218X(98)00146-2},
author = {Robert W. Irving and David F. Manlove}
}

\end{document}